\documentclass[11pt,a4paper]{article}

\usepackage[english]{babel}
\usepackage[utf8]{inputenc}
\usepackage[T1]{fontenc}

\usepackage[a4paper,left=3cm,right=3cm,top=3cm,bottom=3cm]{geometry}

\usepackage{amsmath,amsfonts,amsthm,amssymb,mathrsfs,dsfont} 
\usepackage{mathtools}
\usepackage{graphicx}
\usepackage{subfigure}
\usepackage[colorinlistoftodos]{todonotes}
\hypersetup{
    colorlinks = true,
    allcolors = blue
}
\usepackage[capitalize,nameinlink]{cleveref}
\usepackage{comment}
\usepackage{enumerate}
\usepackage[font=it,labelfont=bf]{caption}

\usepackage{csquotes}

\usepackage{algorithm}
\usepackage{algorithmicx}
\usepackage[noend]{algpseudocode}
\makeatletter
\newcommand*{\algolabel}[2]{%
  \hypertarget{#1}{}%
  \NR@gettitle{#1}%
  \label{#2}%
}
\makeatother

\newtheorem{theorem}{Theorem}[section]
\newtheorem{lemma}[theorem]{Lemma}
\newtheorem{proposition}[theorem]{Proposition}
\newtheorem{corollary}[theorem]{Corollary}

\newtheorem{question}{Question}
\newtheorem{definition}{Definition}[section]

\newtheorem{remark}{Remark}
\newtheorem{claim}{Claim}

\newtheorem*{claim-non}{Claim}
\newtheorem*{subclaim-non}{Subclaim}

\newcommand{\cP}{\mathcal{P}}

\newcommand{\cQ}{\mathcal{Q}}
\newcommand{\cC}{\mathcal{C}}

\newcommand{\cT}{\mathcal{T}}

\newcommand{\bP}{\mathbb{P}}
\newcommand{\bE}{\mathbb{E}}
\renewcommand{\emptyset}{\varnothing}

\renewcommand{\Pr}{\bP}
\newcommand{\E}{\bE}      
\newcommand{\poisson}{\mathsf{Pois}}

\newcommand{\binomial}{\mathsf{Bin}}

\newcommand{\eps}{\varepsilon} 
\newcommand{\desc}[1]{\mathsf{desc}(#1)} 
\newcommand{\anc}[1]{\mathsf{anc}(#1)} 
\newcommand{\RP}[1]{^{(#1)}} 
\newcommand{\hP}[1]{ \hat{\mathcal{P}}_{#1} } 
\newcommand{\sC}{\mathscr{C}} 
\newcommand{\aC}{ \mathscr{C}_{\text{algo}}  }
\newcommand{\iP}{\cP_{\mathsf{indiv}}}

\title{Efficient Reconstruction of Stochastic Pedigrees}
\author{
Younhun Kim\footnote{Massachusetts Institute of Technology. Department of Mathematics. Email: \url{younhun@mit.edu}.} \and
Elchanan Mossel\footnote{Massachusetts Institute of Technology. Department of Mathematics and IDSS. Email: \url{elmos@mit.edu}. Partially supported by awards ONR N00014-16-1-2227, NSF CCF1665252
and DMS-1737944.} \and 
Govind Ramnarayan\footnote{Massachusetts Institute of Technology. CSAIL. Email: \url{govind@mit.edu}. Partially supported by awards NSF CCF1665252 and DMS-1737944.} \and
Paxton Turner\footnote{Massachusetts Institute of Technology. Department of Mathematics. Email: \url{pax@mit.edu}.}
}

\begin{document}
\maketitle

\begin{abstract}
We introduce a new algorithm called {\sc Rec-Gen} for reconstructing the genealogy or \textit{pedigree} of an extant population purely from its genetic data. We justify our approach by giving a mathematical proof of the effectiveness of {\sc Rec-Gen} when applied to pedigrees from an idealized generative model that replicates some of the features of real-world pedigrees. Our algorithm is iterative and provides an accurate reconstruction of a large fraction of the pedigree while having relatively low \emph{sample complexity}, measured in terms of the length of the genetic sequences of the population. We propose our approach as a prototype for further investigation of the pedigree reconstruction problem toward the goal of applications to real-world examples. As such, our results have some conceptual bearing on the increasingly important issue of genomic privacy. 
\end{abstract}
\section{Introduction}

\subsection{Motivation}

The decreased costs of sequencing technologies have enabled large-scale, data-driven analyses of genomes \cite{humangenome}. Recent science and news articles feature stories
only possible due to this plethora of data, such as the recent identification and capture of 
a high-profile criminal \cite{goldenstate} 
predicated on DNA evidence. 
In this effort, an individual's genetic information was compared to a large, curated database called GEDMatch consisting of over one million individual genomes.
In comparison, there exist databases which are of several orders of magnitude larger in size such as MyHeritage ($\sim$3.7 million \cite{myheritageinfo}), 23andMe ($\sim$10 million \cite{23andmeabout}), and Ancestry ($\sim$15 million \cite{ancestry2019press}).


This raises the question: how much kinship information can be learned from DNA? Current databases already contain a considerable amount of this information. Indeed, it is estimated that a given US individual of European ancestry, on average, has a third cousin or closer who is already in the MyHeritage database~\cite{Erlich690}. However, such databases are still far from complete. This calls into question the ability to detect missing kinships based on individuals already present in the database.


This discussion also highlights the issue of \textit{genomic privacy}. Indeed, it becomes much easier to identify and locate individuals by combining the genetic and genealogical information with outside information (addresses, e-mails, family photos, etc.). 
This potential, having already been demonstrated by the resolution of the aforementioned criminal case, was brought to attention by \cite{Erlich690}.
From this point of view, the ability to reconstruct genealogies from collected genetic data is of concern for individuals whose information is revealed, even if one has \emph{never} been sequenced.
Since our work establishes a positive result in a pessimistic scenario where we start with no ground truth information, we believe that our work brings to attention this critical issue via a theoretical framework.

\subsection{Our contributions}

Without any prior knowledge about the ground truth, can we learn \textit{everyone's} genealogy using their genetic information?
In this paper, we study the inference problem of recovering ancestral kinship relationships of a population of \textit{extant} (present-day) individuals, using only their genetic data.
Our goal is to use this extant genetic data to recover the \textit{pedigree} of the extant population, under an idealized model.
A pedigree is a graph whose nodes (individuals) have edges that encode parent-sibling relationships.
The topology and reconstruction of pedigrees are well-studied in bioinformatics from both a theoretical and empirical perspective, and in general the study of pedigrees poses formidable computational and statistical challenges. 

In this paper, we introduce a novel recursive algorithm {\sc Rec-Gen} for pedigree reconstruction. To demonstrate the effectiveness of our approach, we give a mathematical proof that for an idealized generative model on pedigrees, our algorithm is able to approximately recover the true, unknown pedigree only using the genetic data of the extant population. In terms of \textit{sample complexity}, which for our purposes refers to the common gene sequence length of an extant individual, our algorithm greatly outperforms the naive reconstruction method (estimate pairwise distances between the extant individuals, then construct the pedigree that produces these distances). We propose our approach in this work as a prototype for the future study of more general pedigrees, including those involving real-life genetic data, from both a theoretical and empirical perspective. For further discussion on our model of pedigree generation, as well as its features and limitations, see~\cref{sec:model-description} and \cref{sec:model-discussion}.


\subsection{Related works}


A common method in theoretical evolutionary biology is to model lineages and inheritance via a family of directed acyclic graphs.
One line of work is that of \textit{phylogenetics} (refer to \cite{semple2003} for an overview) which uses trees to model the occurrence of large-scale \textit{speciation events} in evolutionary biology.
Another line of work is \textit{coalescent theory}, which focuses on variable-height inheritance trees between genes as its main statistic to infer large-scale \textit{population sizes}, as in e.g. \cite{kim2019}.
In contrast, pedigrees capture small-scale \textit{individual genealogies} that encode familial relationships.
Specifically, most pedigree models are for human genealogies, where we designate exactly two parents to each individual.
By construction, such graphs are no longer trees and warrant different strategies for inference.

\cite{steel2006reconstructing} posed the formal definition of pedigrees using graph-theoretic language. 
In that work, the authors gave combinatorial arguments proving that one can reconstruct complete pedigrees, assuming the correct ancestral history is provided as an input for each extant individual. Our definition of pedigrees is essentially the same as the one outlined by these authors, though we make the simplification that we do not identify the vertex set bipartition (corresponding to the biological sex of the individuals).


To tie in more closely with real-world applications,
one must consider the challenge of estimating these histories from data. Along these lines, \cite{thatte2008reconstructing} studied stochastic processes that one can associate with the pedigree, in such a way that one can prove negative results (information-theoretic impossibility) or positive results (an algorithm) for the reconstruction of the pedigree from extant data.
The stochastic process used to show their positive result was based on a very specific family of Markov chains which allows for inference but is quite different from our model.

For the problem of performing pedigree reconstruction on real data, there is a wealth of literature \cite{thompson00,kirkpatrick2011pedigree,IPED,thompson2013identity,IPED2,shem2014historical,huisman2017pedigree,wang2019pedigree}.
Such studies apply heuristics that take into account various complications and phenomena observed in human genomes, such as varying levels of correlations between different sites and the presence of mutations that are not inherited from parents.

One line of work particularly relevant to this paper is \cite{IPED,IPED2} in which the authors also tackle the problem of pedigree reconstruction from real extant genetic data. Assuming answers to queries of the form, ``how much DNA did $i$ and $j$ simultaneously inherit from their ancestors?'', they design a statistical test that distinguishes between siblings, half-siblings and cousins. 
Their method leverages this information with a maximal-clique finding algorithm to iteratively reconstruct the parents, layer-by-layer. 
There is no proof of correctness provided, but they provide benchmarks on real and simulated data to provide experimental justification. 
Our contributions have a slightly different flavor: using a similar iterative strategy but with a different statistical test (the novel part of our algorithm) and for a more optimistic set of assumptions, one can actually \textit{provably} reconstruct the pedigree correctly in a sample-efficient way, in an asymptotic sense.

The authors of \cite{IPED2} specifically emphasize their method's ability to reconstruct half-siblings. 
Technically speaking, this is not allowed in our model and therefore it may appear to the reader that there is something too restrictive or suboptimal about our analysis.
One major difference between our model and the aforementioned work is that we model \textit{haploid} individuals (one copy of DNA), while in reality humans are \textit{diploids} (two copies of DNA).
Furthermore, in our proof, we guarantee reconstruction of monogamous \textit{couples} of haploid individuals -- in other words, up to permutation of the two individuals within each couple.
It can be observed that given a monogamous pedigree with a haploid model, one can construct a natural, non-monogamous pedigree with a diploid model such that the total variation of the extant data of the two pedigrees is zero.
Therefore, we think that our results should also hold for a diploid model with minor modifications and have correctness guarantees to match the empirical results of the aforementioned work \cite{IPED2}, for example by interpreting \cref{fig:example-1a} as a pair of diploid half-siblings.

Our work is also closely related to the problem of phylogenetic reconstruction~\cite{ErdosSteel99,Mossel04,MosselRoch05,DMR06}. In this setting, symbols are passed from the root of a phylogenetic tree to descendants via a Markov process such as in the Cavender--Farris--Neyman model, a basic model for mutations. Similar to our inference problem in this work, in phylogenetic reconstruction, one is tasked with reconstructing the tree given only the symbols at the leaves. The main result of \cite{ErdosSteel99} characterizes the \textit{sample complexity}---the minimal string length of the data at the leaves such that reconstruction is possible---as logarithmic in the depth of the tree, a phenomenon that our results suggest also holds for the pedigree reconstruction problem. The work \cite{MosselRoch05} provides theoretical guarantees for the problem of learning the phylogenetic generative model (\textit{i.e.}, the topology of the tree as well as the transition matrices), which includes hidden Markov models as a special case, from the extant data under a spectral assumption on the transition matrices (see also later work of \cite{hsu2012spectral}). Most closely related to our approach in this paper is the work \cite{Mossel04}, which shows how to recursively reconstruct phylogenies using techniques from the theory of broadcast processes on trees (see also \cite{DMR06}). This approach provides inspiration for our main algorithm {\sc Rec-Gen}, which uses similar techniques to recursively reconstruct pedigrees. We direct the reader to \cite{evans2000} and \cite{Mossel01} for studies of broadcast processes on trees with binary and large alphabet respectively, and \cite{makur2018broadcasting} for a generalization to directed acyclic graphs.

\subsection{Model description and results}
\label{sec:model-description}
We now give an informal, detailed description of our framework for pedigree reconstruction, with a more detailed treatment of the generative model in \cref{sec:prelims}.
Our generative model on pedigrees consists of two parts: a parametric model for generating the network structure on the set of ancestors and extant individuals, and an inheritance procedure for transmitting genetic data from the \textit{founders}, the oldest individuals in the pedigree, to the extant population. 

To generate the pedigree network structure, we begin with a large founding population of size $N_T$. The founders randomly mate monogamously, and each couple gives birth to a random number of children, so that the average number of offspring per couple is a constant\footnote{More precisely, each couple has a random number of children distributed as a Poisson random variable with expectation $\alpha$.} $\alpha$. This procedure of random monogamous mating continues for $T$ subsequent generations, eventually yielding the extant nodes and a pedigree $\cP$ formed by the individuals in generations $0, 1, \ldots, T$, with $N_i$ nodes at each level $i$. 

Next we describe how genetic data transmits from the founding population to the extant. Every individual in the pedigree has a gene sequence consisting of $B$ symbols placed in $B$ distinct blocks. Each individual in the founding population is initialized with independent uniformly random draws from a very large alphabet $\Sigma$. Now we state how parents pass down genes to their children. In a given block, a child inherits, with equal probability, either its mother's or its father's symbol in the corresponding block. This procedure repeats for all couples in a given generation and then continues over subsequent generations so that genetic data is iteratively transferred through the pedigree, eventually giving rise to the gene sequences of the extant individuals. 

Our main result is summarized in the following theorem. See \cref{thm:main-formal} for a formal statement. 

\begin{theorem}[Main result, informal]
    \label{thm:informal_recon}
    Let $\alpha$ and $\beta$ denote sufficiently large absolute constants independent of $N_T$, the size of the founding population. Let $\eps$ denote a sufficiently small absolute constant independent of $N_T$. Assume that the alphabet size $|\Sigma|$ is very large with respect to $N_T$. 
    
    
    
    Then given extant genetic data produced from the generative model with alphabet $\Sigma$, growth rate $\alpha$, gene sequence length $B = \beta \log N_T $, and number of generations $T = \eps \log N_T$ as described above, the algorithm {\sc Rec-Gen} recovers $90 \%$ of the true pedigree in every generation, with high probability. Moreover, this algorithm runs in polynomial time in the size of the pedigree and the number of blocks per extant individual. 
\end{theorem}

Let $\cP$ denote the true, unknown pedigree. Our formal version of \cref{thm:informal_recon} (see \cref{thm:main-formal}) implies that with high probability {\sc Rec-Gen} outputs a reconstructed pedigree $\hat \cP$ whose size is at least $0.9 N_i$ in each generation $i \in \{0, \ldots, T\}$, such that every node $\hat u \in \hat \cP$ can be identified with exactly one node $u \in \cP$, and this identification preserves relationships in the sense that $\hat u$ is a child of $\hat v$ in $\hat \cP$ if and only if $u$ is a child of $v$ in $\cP$. In graph-theoretic terminology, our reconstruction $\hat \cP$ is a (very large) induced subgraph of the truth $\cP$.

We note that the stipulation that we recover 90\% of the nodes at each level is actually a simplification; in fact, we can make the fraction of reconstructed nodes in each generation \emph{arbitrarily large} by taking $\alpha$ to be large enough. We refer the reader to \cref{thm:main-formal} for details.

\subsection{The {\sc Rec-Gen} algorithm}
\label{sec:rec-gen}

The algorithm {\sc Rec-Gen} consists of a recursive procedure that uses only the genetic information from the extant population to construct a good approximation for the true pedigree $\cP$ of depth $T$ that generated the observations. In the first phase of recursion, the algorithm reconstructs the parents of the extant nodes, which we label as the $1^{\mathrm{st}}$ generation. In the $t^{\mathrm{th}}$ phase, the algorithm adds a $t^{\mathrm{th}}$ generation to the partially reconstructed version of the true pedigree given by the output of the previous phase. The algorithm terminates after $T$ phases of recursion, producing a pedigree $\hat \cP$ with $T$ generations that well-approximates the true, unknown pedigree $\cP$.

We next give a simplified version of our recursive procedure that serves to illustrate the main ideas. See \cref{sec:reconstruction} for a detailed description of {\sc Rec-Gen}. Suppose that we have constructed a pedigree $\hat{\cP}_t$ of depth $t$, and recall that $B$ refers to the length of the gene sequence of an individual. Also recall that a \textit{couple} refers to a pair of mated individuals. 

Note that the first step of our recursive procedure equips each couple with an empirical gene sequence of length B where each block can contain \textit{two} distinct symbols. This empirical gene sequence is constructed based on extant data and should be thought of as determining which symbols belong to at least one of the individuals from the couple in a given block. Also, we say that three gene sequences $\sigma, \sigma', \sigma''$ \textit{overlap} in a block if all three sequences have some symbol in common in that block.


Perform the following steps to output a pedigree $\hat \cP_{t+1}$ of depth $t+1$.


\begin{itemize}
    \item[(1)] {\sc Collect-Symbols} For each couple $c$ in generation $t$ of $\hat \cP_t$, use the extant genetic data to recover symbols that belong to $c$ as follows.
    \begin{itemize}
        \item Recover a symbol $\sigma$ in block $b \in [B]$ of $c$ if $c$ has three extant descendants descended from distinct children of $c$ that all share symbol $\sigma$ in block $b$. 
        \item Repeat this procedure to recover at most one other symbol $\sigma' \neq \sigma$ for $c$ in block $b$.
    \end{itemize}
    \item[(2)] {\sc Test-Siblinghood} For every triple of couples $c, c',c'' \in \hat \cP_t$ in generation $t$, determine $c,c',c''$ to be (mutually) `siblings' if and only if at least $0.21B$ of their recovered symbols mutually overlap. 
    \item[(3)] {\sc Assign-Parents} For every maximal collection $\mathcal{C} =\{ c_1, c_2, \ldots, c_k\}$ of couples in generation $t$ such that every triple in $\cC$ consists of mutual siblings, construct a pair of parents in generation $t+1$ that have as children precisely one individual from each couple in $\cC$.\footnote{We perform this step in such a way that every child is assigned at most $2$ parents.}
\end{itemize}

After $T$ iterations of the above recursive procedure, we output a pedigree $\hat \cP_T$ that gives a good approximation to the underlying pedigree that generated the extant genetic data as described in \cref{thm:informal_recon}. We remark that working with triples as above greatly simplifies our analysis, as discussed in~\cref{sec:triples}. 


\subsection{Model discussion and future directions}
\label{sec:model-discussion}
Our generative model imposes various constraints on the typical pedigrees that we consider. We discuss these modeling assumptions here and also consider the problem of investigating more general models that could more accurately capture properties of real-world data.

First, we consider the assumption that the size of the alphabet $\Sigma$ is very large with respect to the size $N_T$ of the founding population. 
Since a ``block'' represents the unit of inheritance from a parent\footnote{Using biology terminology, each block can be considered as an idealized abstraction of a collection of \emph{single-nucleotide polymorphisms} (sites of variation) with high \emph{linkage disequilibrium} (empirical measure of correlation) that are passed from parent to child.}, this implies that with very high probability all of the founders have distinct symbols in their gene sequences, and no two founders share a common symbol.\footnote{Mathematically, this can be thought of as an improper prior on a countably infinite alphabet $\Sigma$.}
Our large alphabet assumption is equivalent to the assertion that the founders are unrelated.

Second, the stochastic process describing inheritance in our model has the following biological interpretation.
A standard concept in population genetics refers to long-running sequence matches as being \textit{identical by descent} (IBD) if they arose due to inheritance from a common ancestor \cite{thompson2013identity}.
In contrast, the term \textit{identity by state} refers to the event that two identical tracts in the genome arose by coincidence -- via mutations -- in two unrelated individuals. 
Our inheritance model contains the assertion that each block corresponds to true IBD sequences: if two individuals have the same symbol, we can always identify a common ancestor that gave rise to these symbols.

Third, we recall the hypothesis that every couple has on average $\alpha$ children, where $\alpha$ is a sufficiently large absolute constant independent of the size $N_T$ of the founding population. This ensures that, roughly speaking, every new generation is a factor $\alpha/2$ larger than the previous one. Assuming roughly uniform growth of generations, it is necessary that $\alpha > 0\, \,$--- otherwise the population would die out and there would be no extant nodes after $T$ generations. More subtly, it is necessary that $\alpha \geq 2$\, \,--- otherwise, via standard results from the theory of branching processes (see, \textit{e.g.} \cite{KimAxe15}) a founding node has a very low probability of passing on its symbols to the extant. In this situation, even \textit{detection} of such an ancestor from extant genetic data alone is information-theoretically impossible. On the other hand, our assumption that $\alpha$ is a large constant essentially amplifies the signal sent from a founder to the extant, and this simplifies our mathematical analysis.

Our first open question considers relaxing the previously discussed assumptions.

\begin{question}
What theoretical guarantees can be established for pedigree reconstruction in the context of our generative model when $\alpha$ is very close to $2$? What about when the size of the alphabet $\Sigma$ is finite? Can we analyze more generic models of inheritance where blocks are not inherited i.i.d. from parents?
\end{question}

A more subtle consequence of our generative model is \textit{inbreeding}, a term we use to refer to the following phenomena: (1) the presence of multiple lowest common ancestors for a pair of extant nodes, and (2) the presence of mated couples such that the two individuals in the couple have a lowest common ancestor (LCA) (see \cref{def:LCA} for the formal definition of an LCA). The \textit{degree} of inbreeding qualitatively refers to the frequency of such structures in the pedigree. Moreover, inbreeding as in (2) is mathematically equivalent to having cycles in the pedigree.
In general, a higher degree of inbreeding makes the pedigree reconstruction problem more difficult and in some cases information-theoretically impossible (see \cref{sec:examples} for detailed examples). Our choice of model allows for some degree of inbreeding, and our algorithm and analysis are carefully tailored to circumvent this obstacle.

Other assumptions inherent in our model include that the pedigree is \textit{graded}, \textit{i.e.}, couples are formed from individuals in the same generation, and \textit{monogamous}: a given individual only mates with one other individual. Furthermore, \textit{mutations} --- errors in the transmission of genetic data from parents to offspring --- are a central component in biological applications that our current model does not incorporate. 

\begin{question}
What theoretical guarantees can be established for reconstruction of pedigrees in generative models with some combination of (i) a higher degree of inbreeding, (ii) mutations, (iii) non-monogamous mating, and (iv) inter-generational mating?
\end{question}

\section{Inference challenges and techniques}
\label{sec:techniques}
In this section, we detail some of the challenges posed by the reconstruction of pedigrees constructed from our generative model as well as our techniques and analysis for handling them. To develop some intuition for our strategy, we first illustrate some of the properties of pedigrees using concrete examples.

\subsection{Examples: complications from inbreeding}
\label{sec:examples}



Recall that two individuals $u,v$ that share the same set of parents are \textit{siblings}. 
If two individuals share a common subset of grandparents (but not parents), we refer to them as \textit{cousins}.

\begin{figure}
    \centering
    \subfigure[three sets of grandparents (cousins, one way)]{
    \label{fig:example-1a}
    \begin{tikzpicture}[scale=0.8,
        > = stealth,
        auto,
        node distance = 3cm, 
        semithick  
    ]
    
	\tikzset{rect state/.style={draw,rectangle,thick,minimum size=4mm}}
	\tikzset{circ state/.style={draw,circle,thick,minimum size=8mm}}
	\tikzset{dashed state/.style={draw,circle,dashed,thick,minimum size=4mm}}
	
	\node[circ state] (1A) at (-1.675,0) {$\boxed{k}$};
	\node[circ state] (1B) at (1.675,0) {$\boxed{\ell}$};
	
	\node[circ state] (2A) at (-3.6375,2) {$\boxed{g}$};
	\node[circ state] (2B) at (-2.0625,2) {$\boxed{h}$};
	\node[circ state] (2C) at (2.0625,2) {$\boxed{i}$};
	\node[circ state] (2D) at (3.6375,2) {$\boxed{j}$};
	
	\node[circ state] (3A) at (-6.3375,4) {$\boxed{a}$};
	\node[circ state] (3B) at (-4.7625,4) {$\boxed{b}$};
	\node[circ state] (3C) at (-0.825,4) {$\boxed{c}$};
	\node[circ state] (3D) at (0.825, 4) {$\boxed{d}$};
	\node[circ state] (3E) at (4.7625,4) {$\boxed{e}$};
	\node[circ state] (3F) at (6.3375,4) {$\boxed{f}$};
	
	\path[->] (3A) edge (2A);
	\path[->] (3B) edge (2A);
	\path[->] (3C) edge (2B);
	\path[->] (3D) edge (2B);
	\path[->] (3C) edge (2C);
	\path[->] (3D) edge (2C);
	\path[->] (3E) edge (2D);
	\path[->] (3F) edge (2D);
	\path[->] (2A) edge (1A);
	\path[->] (2B) edge (1A);
	\path[->] (2C) edge (1B);
	\path[->] (2D) edge (1B);
    \end{tikzpicture}
    }
    ~
    \subfigure[two sets of grandparents (cousins, two ways)]{
    \label{fig:example-1b}
    \begin{tikzpicture}[scale=0.8,
        > = stealth,
        auto,
        node distance = 3cm, 
        semithick  
    ]
    
	\tikzset{rect state/.style={draw,rectangle,thick,minimum size=4mm}}
	\tikzset{circ state/.style={draw,circle,thick,minimum size=8mm}}
	\tikzset{dashed state/.style={draw,circle,dashed,thick,minimum size=4mm}}
	
	\node[circ state] (1A) at (-1.675,0) {$\boxed{k}$};
	\node[circ state] (1B) at (1.675,0) {$\boxed{\ell}$};
	
	\node[circ state] (2A) at (-3.6375,2) {$\boxed{g}$};
	\node[circ state] (2B) at (-2.0625,2) {$\boxed{h}$};
	\node[circ state] (2C) at (2.0625,2) {$\boxed{i}$};
	\node[circ state] (2D) at (3.6375,2) {$\boxed{j}$};
	
	\node[circ state] (3A) at (-3.6375,4.5) {$\boxed{a}$};
	\node[circ state] (3B) at (-2.0625,4.5) {$\boxed{b}$};
	\node[circ state] (3C) at (2.0625,4.5) {$\boxed{c}$};
	\node[circ state] (3D) at (3.6375,4.5) {$\boxed{d}$};
	
	\path[->] (3A) edge (2A);
	\path[->] (3B) edge (2A);
	\path[->] (3C) edge (2B);
	\path[->] (3D) edge (2B);
	\path[->] (3A) edge (2C);
	\path[->] (3B) edge (2C);
	\path[->] (3C) edge (2D);
	\path[->] (3D) edge (2D);
	\path[->] (2A) edge (1A);
	\path[->] (2B) edge (1A);
	\path[->] (2C) edge (1B);
	\path[->] (2D) edge (1B);
    \end{tikzpicture}
    }

    \caption{
    Simple examples of depth-3 complete pedigrees with a single block. The letters inside the boxes represents the block data.
    \ref{fig:example-1a}: The overlap probability is $\Pr(k = \ell) = \tfrac{1}{8}$.
    \ref{fig:example-1b}: An altered version of~\ref{fig:example-1a} with only two sets of grandparents, which yields $\Pr(k = \ell) = \frac{1}{4}$.
    }
    \label{fig:example-1}
\end{figure}
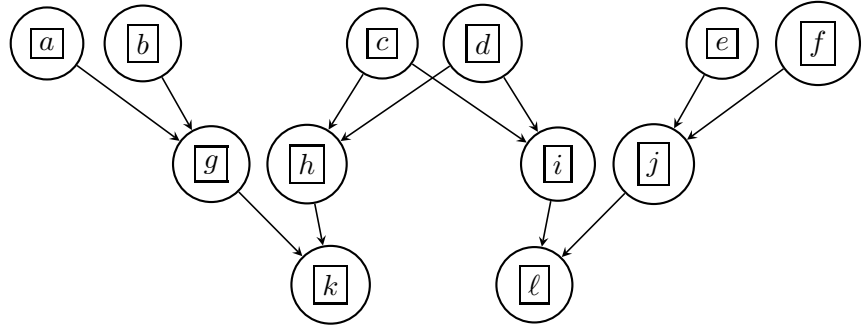
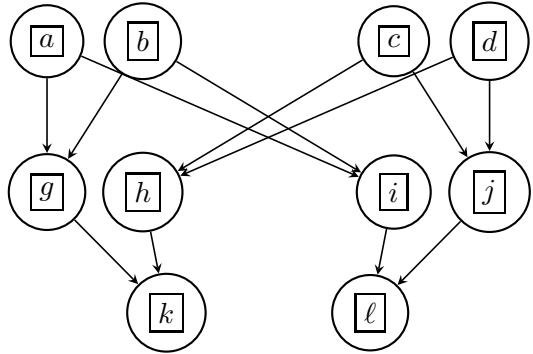

First consider the pedigrees displayed in~\cref{fig:example-1a}. An important statistic for determining relationships is the correlation between symbols of nodes at the same level. Consider the event $E$ that the left extant shares the same symbol as the right extant. Note that these two extant nodes are cousins sharing a single set of grandparents. The grandparents are the founders in this example, so we assign to each of them a unique symbol ($a \neq b \neq c \neq d \neq e \neq f$).
The occurrence of $E$ implies that $k = c$ or $k = d$ via the left extant receiving a symbol from its right parent; this occurs with probability $\tfrac{1}{2}$.
Conditioned on this occurring, the right extant block $\ell$ is the same as $k$ with probability $\frac{1}{4}$, so the overall probability that both receive the same symbol is $\frac{1}{8}$.

Compare this to the example shown in \cref{fig:example-1b}, where the two extant are cousins in two ways (\textit{siblings marrying siblings}).
Note that whichever symbol (out of $a,b,c,d$) that $k$ is, the right grandchild receives the same independently with probability $\tfrac{1}{4}$. This is an example of a type of inbreeding where two extant nodes have more than one LCA.

\begin{figure}
    \centering
    \subfigure[four siblings begetting cousins.]{
    \label{fig:example-2a}
    \begin{tikzpicture}[scale=0.8,
        > = stealth,
        auto,
        node distance = 3cm, 
        semithick  
    ]
    
	\tikzset{rect state/.style={draw,rectangle,thick,minimum size=4mm}}
	\tikzset{circ state/.style={draw,circle,thick,minimum size=8mm}}
	\tikzset{dashed state/.style={draw,circle,dashed,thick,minimum size=4mm}}
	
	\node[circ state] (1A) at (-1.675,0) {$\boxed{k}$};
	\node[circ state] (1B) at (1.675,0) {$\boxed{\ell}$};
	
	\node[circ state] (2A) at (-3,2) {$\boxed{e}$};
	\node[circ state] (2B) at (-1.7,2) {$\boxed{f}$};
	\node[circ state] (2C) at (1.7,2) {$\boxed{g}$};
	\node[circ state] (2D) at (3,2) {$\boxed{h}$};
	
	\node[circ state] (3A) at (-1.0,4) {$\boxed{a}$};
	\node[circ state] (3B) at (1.0,4) {$\boxed{b}$};
	
	\path[->] (3A) edge (2A);
	\path[->] (3B) edge (2A);
	\path[->] (3A) edge (2B);
	\path[->] (3B) edge (2B);
	\path[->] (3A) edge (2C);
	\path[->] (3B) edge (2C);
	\path[->] (3A) edge (2D);
	\path[->] (3B) edge (2D);
	\path[->] (2A) edge (1A);
	\path[->] (2B) edge (1A);
	\path[->] (2C) edge (1B);
	\path[->] (2D) edge (1B);
    \end{tikzpicture}
    }  \hspace{2.5cm} \subfigure[two siblings begetting siblings.]{
    \label{fig:example-2b}
    \begin{tikzpicture}[scale=0.8,
        > = stealth,
        auto,
        node distance = 3cm, 
        semithick  
    ]
    
	\tikzset{rect state/.style={draw,rectangle,thick,minimum size=4mm}}
	\tikzset{circ state/.style={draw,circle,thick,minimum size=8mm}}
	\tikzset{dashed state/.style={draw,circle,dashed,thick,minimum size=4mm}}
	
	\node[circ state] (1A) at (-1,0) {$\boxed{k}$};
	\node[circ state] (1B) at (1,0) {$\boxed{\ell}$};
	
	\node[circ state] (2A) at (-1.0,2) {$\boxed{e}$};
	\node[circ state] (2B) at (1.0,2) {$\boxed{f}$};
	
	\node[circ state] (3A) at (-1.0,4) {$\boxed{a}$};
	\node[circ state] (3B) at (1.0,4) {$\boxed{b}$};
	
	\path[->] (3A) edge (2A);
	\path[->] (3B) edge (2A);
	\path[->] (3A) edge (2B);
	\path[->] (3B) edge (2B);
	\path[->] (2A) edge (1A);
	\path[->] (2B) edge (1A);
	\path[->] (2A) edge (1B);
	\path[->] (2B) edge (1B);
    \end{tikzpicture}
    }
    
    \caption{Two examples of complete pedigrees with inbreeding. The extants in \ref{fig:example-2a} are cousins, yet they have a coincidence of $\tfrac{1}{2}$ as if they were generic siblings from unrelated parents.
    In comparison, \ref{fig:example-2b} yields $\tfrac{3}{4}$ which exceeds the coincidence of siblings.
    }
    \label{fig:example-inbreeding}
\end{figure}
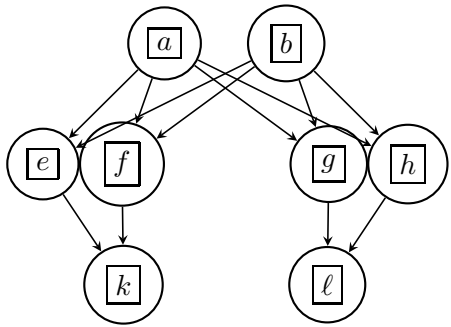
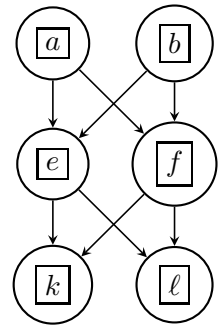


The examples in \cref{fig:example-inbreeding} demonstrate how the correlation between extant nodes is boosted due to the presence of inbreeding. Note that in the \textit{generic} case where extant siblings have an ancestral pedigree that is a tree, these individuals have a $\tfrac{1}{2}$ fraction overlap in their blocks. 
For comparison, let us compute the probability of coincidence for the two extant nodes in \cref{fig:example-2a}.
The probability that $k=a$, for example, is
\[
    \Pr(k = a)
    = \Pr(e = f = a) + \frac{1}{2} \Pr\left(\{e,f\} = \{a,b\}\right) 
    = \frac{1}{4} + \left( \frac{1}{2} \right)^2 
    = \frac{1}{2}.
\]
Since $k$ and $\ell$ inherit symbols independently from their grandparents, the overall probability is
\[
    \Pr(k = \ell) 
    = \Pr(k = \ell = a) + \Pr(k = \ell = b) 
    = \left( \frac{1}{2} \right)^2 + \left( \frac{1}{2} \right)^2
    = \frac{1}{2},
\]
which is precisely the probability that two generic siblings inherit the same symbol. 

The situation in \cref{fig:example-2b} is even more pronounced. 
The two parents share the same symbol (either $a$ or $b$) with probability $\tfrac{1}{2}$ and have different symbols with probability $\tfrac{1}{2}$.
This means that the coincidence probability is now $\tfrac{1}{2} + \tfrac{1}{2} \times \tfrac{1}{2} = \tfrac{3}{4}$: their correlation between overlaps is much stronger than that of siblings in the generic case.

From the example in \cref{fig:example-2a}, we conclude that the statistical model of extant data parametrized by pedigrees is unidentifiable. 
Stated another way, it is information-theoretically impossible to distinguish between siblings and inbred cousins using only extant data.
Thus, in order for any algorithm to succeed in reconstructing a large fraction of the pedigree using only extant data, it is necessary to bound the amount of inbreeding in the ensemble of pedigrees of interest. We accomplish this using a careful analysis of our generative model. 




\subsection{Informal analysis of {\sc Rec-Gen}}
\label{sec:analysis}

In this section, we present a high-level analysis of the {\sc Rec-Gen} algorithm.~\cref{thm:informal_recon} states that {\sc Rec-Gen} yields an accurate reconstruction on $90\%$ of nodes for typical pedigrees from our generative model\footnote{We note again that the $90\%$ is for simplicity of exposition, and in reality we can recover an arbitrarily large fraction of nodes. This is made precise in \cref{thm:main-formal}.}. Note that a formal statement of this theorem, our main result, is given by~\cref{thm:main-formal}, and a complete proof is contained in the upcoming sections.

Suppose we construct a pedigree $\hat \cP_t$ on $t$ generations that, for simplicity of the discussion, \textit{exactly} matches the true, unknown pedigree $\cP$ up to generation $t$.
We show that {\sc Collect-Symbols}, {\sc Test-Siblings}, and {\sc Assign-Parents} applied to $\hat \cP_t$ provide an accurate reconstruction of $90 \%$ of the nodes at generation $t+1$. 
In the remainder of this section we give a high-level argument that the output $\hat \cP_{t+1}$ satisfies the following conditions:
\begin{itemize}
    \item[(i)] every individual $\hat u$ in $\hat \cP_{t+1}$ can be identified with a unique individual $u$ in $\cP$ at generation $t+1$,
    \item[(ii)] at most $10\%$ of the nodes in generation $t+1$ of $\cP$ are not identified with an individual in $\hat \cP_{t+1}$, and 
    \item[(iii)] if $v$ is a child of $\hat u$ in $\hat \cP_{t+1},$ then $v$ is a child of $u$ in $\cP$.
\end{itemize}
Recall that for the purposes of reconstruction, we only have access to the genetic data of the extant.



In this discussion, we refer to three couples $c, c',c'' \in \cP$ as (mutual) siblings if there exist individuals $u \in c, u' \in c',$ and $u'' \in c''$ such that $u, u',$ and $u''$ are mutually siblings. A \textit{clique} refers to a collection of couples $\cC = \{c_1, \ldots, c_k\}$ such that every triple from $\cC$ consists of mutual siblings.

The next two facts are essential to the argument. 


\begin{enumerate}
    \item[(A)] If {\sc Collect-Symbols} recovers symbol $\sigma$ in block $b$ for a couple $c$ in generation $t$, then $c$ also has the symbol $\sigma$ in block $b$ in $\cP$  (\cref{clm:symbol_consistency}).
    \item[(B)] {\sc Collect-Symbols} recovers at least $99\%$ of the symbols for at least $99\%$ of the couples in generation $t$ (\cref{lem:awesome-are-b-good}).
\end{enumerate}

Together, (A) and (B) imply that for $99\%$ of the couples in generation $t$, our algorithm gets all of the siblings relationships between these couples correct. 
To see why, we can use a similar calculation as in the first example of \cref{sec:examples} to conclude that the average overlap between the symbols of three individuals that are mutually siblings is $25\%$. 
By concentration of binomial random variables about their means, it follows that with high probability, all triples of individuals that are mutually siblings in $\cP$ have at least $24.9\%$ mutual overlap between their symbols. 
A simple union bound combined with (A) and (B) implies that for most triples of individuals in generation $t$ that are mutually siblings, the recovered symbols from {\sc Collect-Symbols} in those individuals' corresponding couples have overlap at least $21\%$. 
Hence, {\sc Test-Siblinghood} infers correct siblinghood relationships for a majority of triples.

Moreover, our siblings test on the recovered symbols does not have any false-positives:

\begin{itemize}
    \item[(C)] {\sc Test-Siblinghood} never misclassifies non-siblings as siblings (\cref{lem:sibling_graph_reconstruct}).
\end{itemize}

The next and last key fact argues that our naive assignment of parents to individuals in cliques as in {\sc Assign Parents} is in fact the correct assignment in a typical pedigree. This property holds with very high probability over our generative model.

\begin{itemize}
    \item[(D)] Let $\cC \subset \cP$ denote a clique at generation $t$ in the true pedigree. Then there exists a couple $\tilde{c}$, which we refer to as the \textit{parents of} $\cC$, in generation $t+1$ of $\cP$ that has exactly one child in every couple of $\cC$, and no other couple has more than $1$ child in $\cC$ (\cref{lem:clique_has_unique_parent}). 
\end{itemize}

Together, (A), (B), (C), and (D) imply that our reconstruction criteria (i), (ii), and (iii) from the beginning of this section hold, as we now justify. Recall that we already showed (A) and (B) imply that we classify a large fraction of the couples at generation $t$ correctly as siblings. Moreover, part (C) and the transitivity of siblinghood in $\cP$ imply that cliques in our reconstruction really correspond to cliques in the truth. By part (D) such cliques have unique parents. Thus, for (i), we identify newly constructed couples $\hat u \in \cP_{t+1}$ with the \textit{unique} parents $u \in \cP$ of the clique formed by the children of $\hat u$, further pairing the two individuals in $u$ with those in $\hat u$ arbitrarily. With this identification, (iii) follows immediately. To show part (ii), later in the paper we give a sufficient condition for a couple at generation $t$ to have $99\%$ of its symbols collected by {\sc Collect-Symbols} as in (B) (see \cref{lem:awesome-are-b-good}). Then we show that $90 \%$ of individuals in generation $t+1$ have children in such couples (see \cref{prop:many-awesome}), which proves part (ii). Essentially, this sufficient condition amounts to saying that a couple $c$ at generation $t$ has no inbreeding (cycles) above or below it (\textit{i.e.} among its ancestors or descendants, respectively) and that the pedigree of descendants of $c$ contains a $\alpha/4$-ary tree (see \cref{def:awesome_node}). 




\subsection{Motivation for using triples} 
\label{sec:triples}

It is tempting to employ a seemingly simpler recursive scheme than the one described in~\cref{sec:rec-gen} that operates on pairs instead of triples. 
As an example, consider an alternative recursive procedure such that:
\begin{enumerate}
    \item {\sc Collect-Symbols} only uses \textbf{pairs} of extant descendants to recover symbols of a couple $c,$
    \item {\sc Test-Siblinghood} considers only \textbf{pairs} of couples at generation $t$ and detects them to be siblings if their strings overlap by at least $49 \%$, and
    \item {\sc Assign-Parents} assigns parents to individuals in maximal collections $\cC$ such that every \textbf{pair} of couples is (tested as) siblings.
\end{enumerate}
Unfortunately, this simpler approach encounters two major technical complications.

First, working with a pairwise siblings test introduces a problem for the step of assigning parents. 
Define a \textit{pairwise clique} to be a collection of couples so that every pair of couples passes the pairwise siblings test. 
With high probability, it turns out in every generation there exist a constant number of pairwise cliques that are not explained in the naive way of assigning to this clique parents that have precisely one child per couple. 
In particular, in the true pedigree $\cP$ it is possible to have three couples that mutually pass the pairwise siblings test, yet there are \textbf{three} distinct parent couples each having precisely two children among these three couples. See \cref{fig:sib-bad-ex} for an illustration. This type of structure, though rare, occurs a constant number of times in each generation, and thus introduces inherent errors in our reconstruction that accumulate at every step of iteration.

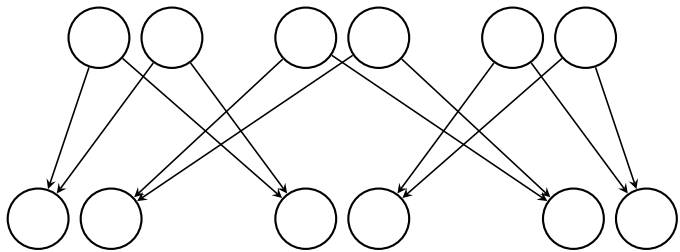
\begin{figure}[!ht]
    \centering
    \begin{tikzpicture}[scale=0.8,
        > = stealth,
        auto,
        node distance = 3cm, 
        semithick  
    ]
    
	\tikzset{rect state/.style={draw,rectangle,thick,minimum size=4mm}}
	\tikzset{circ state/.style={draw,circle,thick,minimum size=8mm}}
	\tikzset{dashed state/.style={draw,circle,dashed,thick,minimum size=4mm}}
	
	\node[circ state] (1A) at (-5.0,2) {};
	\node[circ state] (1B) at (-3.8,2) {};
	\node[circ state] (1C) at (-0.6,2) {};
	\node[circ state] (1D) at (0.6,2) {};
	\node[circ state] (1E) at (3.8,2) {};
	\node[circ state] (1F) at (5.0,2) {};
	
	\node[circ state] (2A) at (-4.0,5) {};
	\node[circ state] (2B) at (-2.8,5) {};
	\node[circ state] (2C) at (-0.6,5) {};
	\node[circ state] (2D) at (0.6,5) {};
	\node[circ state] (2E) at (2.8,5) {};
	\node[circ state] (2F) at (4.0,5) {};
	
	\path[->] (2A) edge (1A);
	\path[->] (2B) edge (1A);
	\path[->] (2A) edge (1C);
	\path[->] (2B) edge (1C);
	\path[->] (2E) edge (1D);
	\path[->] (2F) edge (1D);
	\path[->] (2E) edge (1F);
	\path[->] (2F) edge (1F);
	\path[->] (2C) edge (1B);
	\path[->] (2D) edge (1B);
	\path[->] (2C) edge (1E);
	\path[->] (2D) edge (1E);
    \end{tikzpicture}
    
    \caption{
        An undesirable subpedigree, where three child couples have mutual siblingship, but they do not mutually share a parent couple.
    }
    \label{fig:sib-bad-ex}
\end{figure}

A second problem caused by working with pairs arises in the step of collecting symbols. The pairwise version of our algorithm assigns a symbol to a couple if that symbol occurs in two extant descendants that are descended from distinct children of that couple. In our generative model, it turns out that with high probability there are a \textit{logarithmic} number of pairs of extant nodes that have at least two LCA's. For such pairs, the pairwise algorithm does not accurately assign symbols to their reconstructed ancestors. Similar to the previous issue, these errors snowball and make the analysis for proving \cref{thm:informal_recon} very difficult.

On the other hand, working with an algorithm using triples as described in \cref{sec:rec-gen} makes for a much cleaner analysis and nicer reconstruction guarantee. This innovation circumvents the technical complications of the pairwise version because every clique (recall that this is a collection of couples where every triple consists of mutual siblings) can be explained in a naive way (\cref{lem:clique_has_unique_parent}), and in our generative model every triple of extant individuals descended from distinct children of a given ancestor have that ancestor as their \textit{unique} LCA with very high probability (\cref{lem:unique_joint_LCA}).


\subsection{Outline of technical arguments}

The remainder of the paper, which provides a formal proof of \cref{thm:informal_recon}, is divided into four parts.
\begin{itemize}
    \item \cref{sec:prelims} provides preliminary definitions and a formal definition of our generative model.
    \item \cref{sec:structure} proves important properties about the typical network structure of pedigrees from our generative model.
    \item \cref{sec:symbol-inheritance} proves important properties about the block statistics of the extant nodes in a typical pedigree from our generative model. 
    \item \cref{sec:reconstruction} gives a precise description of {\sc Rec-Gen} and provides a formal statement and proof of \cref{thm:informal_recon}.
\end{itemize}

Specifically, in \cref{sec:structure} we rigorously quantify the degree of inbreeding in typical pedigrees from our model by counting the number of \textit{collisions} (see \cref{def:collision} and \cref{lem:bound-collisions-main}). This has several useful consequences, including that every clique has a unique parent (fact (D) from \cref{sec:analysis}, also see \cref{lem:clique_has_unique_parent})
and that the extant individuals used in {\sc Collect-Symbols} have a unique LCA (see \cref{lem:unique_joint_LCA}). In particular, the latter is key to showing fact (A) from \cref{sec:analysis}. 

In \cref{sec:symbol-inheritance}, we provide a definition (see \cref{def:awesome_node}) that essentially characterizes the individuals in $\cP$ that are reconstructible via {\sc Rec-Gen}. We show that couples involving such individuals, referred to as \textit{awesome couples}, transmit many of their symbols to the extant, with high probability (see \cref{lem:awesome-are-b-good}). In particular, awesome couples have at least $99\%$ of their symbols recovered by {\sc Collect-Symbols} (fact (B) from \cref{sec:analysis}). We also prove an important result for our siblings test: triples of individuals that are not mutually siblings have mutually overlap at most $19\%$ (see \cref{lem:nonsiblings_symbols}). This combined with fact (A) from \cref{sec:analysis} essentially shows that {\sc Test-Siblinghood} never classifies non-siblings as siblings (fact (C) from \cref{sec:analysis}, see also \cref{lem:sibling_graph_reconstruct}).  

Our final section, \cref{sec:reconstruction} ties everything together, following fairly closely the high-level argument presented in \cref{sec:analysis} to prove the formal version of \cref{thm:informal_recon}. 

\section{Preliminaries}
\label{sec:prelims}
\subsection{Key definitions and terms}
\label{sec:definitions}

\begin{definition}
A \textbf{pedigree} $\cP = (V, E)$ is a directed acyclic graph (DAG) with vertices $V$ and edges $E$ where every vertex has indegree at most $2$. The collection of vertices of indegree zero are referred to as the \textbf{founders}, and the collection of vertices of outdegree zero are referred to as the \textbf{extant}.
\end{definition}

\begin{definition}
If the indegree of each vertex in the underlying DAG is either $2$ or $0$, then $\cP$ is called a \textbf{complete} pedigree.
\end{definition}


In this work, we focus on a special family of complete pedigrees that are both \textit{graded} and \textit{monogamous}.

\begin{definition}
$\cP$ is said to be \textbf{graded} if the vertices $V(\cP)$ can be partitioned into $\bigcup_{i=0}^{T} V_i(\cP)$ such that $V_T(\cP)$ are the founders, $V_0(\cP)$ are the extant, and all directed paths $e_T, \ldots, e_1$ from $V_T(\cP)$ to $V_0(\cP)$ can be written as a sequence of edges $e_t = (v_t \to v_{t-1})$ where $v_t \in V_t(\cP)$ and $v_{t-1} \in V_{t-1}(\cP)$ for each $t$.
The founders' index $T$ is the \textbf{depth} of the pedigree.

$\cP$ is said to be \textbf{monogamous} if for every vertex $u$ of outdegree $> 0$, there exists a unique vertex $u'$ such that $(u \to v) \in E \iff (u' \to v) \in E$.
The unordered pair $\{u,u'\}$ is referred to as a \textbf{couple}. 
\end{definition}
We assume that every non-extant individual in the pedigree is in a couple, and so the number of vertices at each non-extant level is even. This assumption is effectively without loss of generality---if an individual is not in a couple, then it has no descendants, and so we cannot recover information about this individual or even its existence.

An example of a complete, graded, monogamous pedigree is shown in \cref{fig:example-1a}.
In our model, symbols are passed down from parents to children in a completely symmetric way. 
Thus, given the data of the children, it is impossible to distinguish the owner of each symbol from amongst the two parents.
The goal of this paper is to show how one can provably infer the structure of a complete pedigree from extant genetic data via the reconstruction of the ancestral symbols, modulo \textit{block phasing} (determining which symbol belongs to which parent for each block).
Therefore, we introduce the following version of a pedigree which condenses this information.


\begin{definition}
\label{def:coupled-pedigree}
A \textbf{coupled} pedigree $\cQ = (V_{\cQ},E_{\cQ})$ induced by a complete, monogamous pedigree $\cP = (V_{\cP}, E_{\cP})$ is defined as follows:
\begin{itemize}
\item $V_{\cQ} \subset \binom{V_{\cP}}{2}$ is obtained by merging couples $c = \{u, u'\} \subset V_{\cP}$ into a single node (extant individuals remain singletons), introducing edge multiplicity.
\item $E_{\cQ}$ is the result of halving the number of resulting copies of each edge after merging couples.
\end{itemize}
\end{definition}

\begin{figure}
    \centering
    \subfigure[coupled version of \ref{fig:example-1a}]{
    \label{fig:example-3a-couples}
    \begin{tikzpicture}[scale=0.8,
        > = stealth,
        auto,
        node distance = 3cm, 
        semithick  
    ]
    
	\tikzset{rect state/.style={draw,rectangle,thick,minimum size=4mm}}
	\tikzset{circ state/.style={draw,circle,thick,minimum size=8mm}}
	\tikzset{dashed state/.style={draw,circle,dashed,thick,minimum size=4mm}}
	
	\node[circ state] (1A) at (-1.675,0) {$\boxed{k}$};
	\node[circ state] (1B) at (1.675,0) {$\boxed{\ell}$};
	
	\node[circ state] (2AB) at (-2,2) {$\boxed{\{g,h\}}$};
	\node[circ state] (2CD) at (2,2) {$\boxed{\{i,j\}}$};
	
	\node[circ state] (3AB) at (-4,4) {$\boxed{\{a,b\}}$};
	\node[circ state] (3CD) at (0.0,4) {$\boxed{\{c,d\}}$};
	\node[circ state] (3EF) at (4,4) {$\boxed{\{e,f\}}$};
	
	\path[->] (3AB) edge (2AB);
	\path[->] (3CD) edge (2AB);
	\path[->] (3CD) edge (2CD);
	\path[->] (3EF) edge (2CD);
	\path[->] (2AB) edge (1A);
	\path[->] (2CD) edge (1B);
    \end{tikzpicture}
    }
    ~
    \subfigure[coupled version of \ref{fig:example-1b}]{
    \label{fig:example-3b-couples}
    \begin{tikzpicture}[scale=0.8,
        > = stealth,
        auto,
        node distance = 3cm, 
        semithick  
    ]
    
	\tikzset{rect state/.style={draw,rectangle,thick,minimum size=4mm}}
	\tikzset{circ state/.style={draw,circle,thick,minimum size=8mm}}
	\tikzset{dashed state/.style={draw,circle,dashed,thick,minimum size=4mm}}
	
	\node[circ state] (1A) at (-1.675,0) {$\boxed{k}$};
	\node[circ state] (1B) at (1.675,0) {$\boxed{\ell}$};
	
	\node[circ state] (2AB) at (-2.5,2) {$\boxed{\{g,h\}}$};
	\node[circ state] (2CD) at (2.5,2) {$\boxed{\{i,j\}}$};
	
	\node[circ state] (3AB) at (-3,4.5) {$\boxed{\{a,b\}}$};
	\node[circ state] (3CD) at (3,4.5) {$\boxed{\{c,d\}}$};
	
	\path[->] (3AB) edge (2AB);
	\path[->] (3AB) edge (2CD);
	\path[->] (3CD) edge[line width=5pt,draw=white] (2AB);
	\path[->] (3CD) edge (2AB);
	\path[->] (3CD) edge (2CD);
	\path[->] (2AB) edge (1A);
	\path[->] (2CD) edge (1B);
    \end{tikzpicture}
    }
    \caption{
    \ref{fig:example-1a} induces coupled pedigree \ref{fig:example-3a-couples}, while \ref{fig:example-1b} induces \ref{fig:example-3b-couples}.
    }
    \label{fig:example-3}
\end{figure}
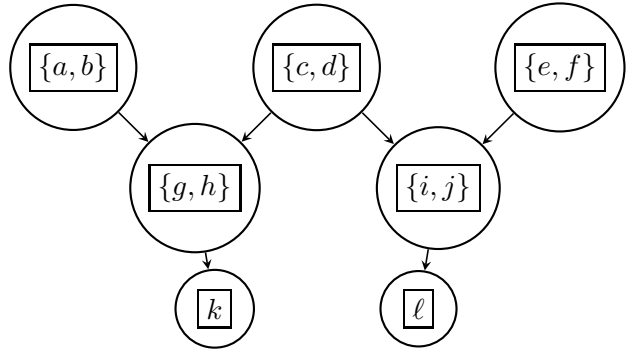
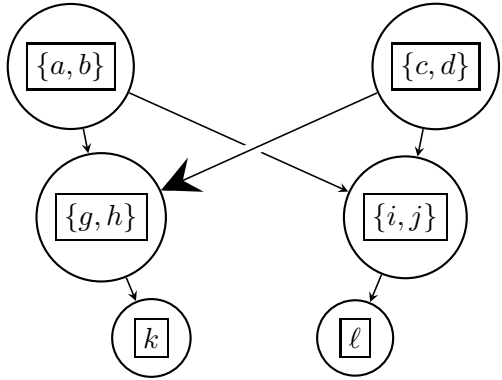

In particular, a coupled pedigree is also a pedigree. Examples are drawn in \cref{fig:example-3} in relation to \cref{fig:example-1}, where the complete pedigree~\ref{fig:example-1a} induces a coupled pedigree~\ref{fig:example-3a-couples} and \ref{fig:example-1b} induces \ref{fig:example-3b-couples}.

The only information that is lost after transforming a complete, monogamous pedigree into a coupled pedigree is the block phasing. 
Indeed, observe that given the coupled structure $\cQ = (V_{\cQ}, E_{\cQ})$, one can easily obtain the individual structure $\cP = (V_{\cP}, E_{\cP})$ up to block phasing as follows: 
(1) add the extant individuals in $V_0 \subset V_{\cQ}$ to $V_{\cP}$, 
(2) for every non-extant node $c \in V_{\cQ}$ add individuals $u_c, u_c'$ to $V_{\cP}$, and 
(3) given parents $c_1$ and $c_2$ of $c$ in $\cQ$, add the four edges $u_{c_1} \to u_c, u_{c_1}' \to u_c, u_{c_2} \to u_c', u_{c_2}' \to u_c'$ to $E_{\cP}$. 
In addition, if $\cP$ is graded, $\cQ$ retains a graded structure $V_{\cQ} = V_0(\cQ) \cup \cdots \cup V_T(\cQ)$ so that $V_0(\cQ)$ are the extant nodes and $V_1(\cQ), \ldots, V_T(\cQ)$ are depth-graded couple nodes. In particular, the graph structure of an individuals pedigree $\cP$ uniquely determines the graph structure of its associated coupled pedigree $\cQ$ and vice versa.

Given the previous discussion, since our goal is to recover the graph structure of an underlying true pedigree $\cP$ given gene sequences of a large number of extant individuals, it suffices to reconstruct the associated coupled pedigree $\cQ$. 


Furthermore, since the graph underlying a pedigree is a DAG, given a subset $S$ of the pedigree, it is natural to consider the notion of ``ancestors'' (nodes $\mathsf{anc}(S)$ \textbf{from} which there is a directed path to $S$) and ``descendants'' (nodes $\mathsf{desc}(S)$ \textbf{to} which there is a directed path from $S$). Also for simplicity, we stipulate that every node $v$ is both a descendant and an ancestor of itself, \textit{i.e.}, $v \in \mathsf{anc}(v)$ and $v \in \mathsf{desc}(v)$.
Since the indegree of each node can be more than one, it is possible for two nodes to have more than one ``lowest common ancestor''. We define this now.

\begin{definition}[Lowest Common Ancestors]
\label{def:LCA}
Let $S$ denote a set of nodes in a pedigree $\cP$. 
The set of \textbf{lowest common ancestors} of $S$, denoted $\mathsf{LCA}(S)$, consists of all nodes $u \in \cP$ such that $u$ is an ancestor of every node in $S$, and moreover, no descendant of $u$ is an ancestor of every node in $S$. 
\end{definition}

During our analysis, we often restrict our attention to the information that the pedigree contains about the ancestors or descendants of a particular collection of nodes.
In particular, we want to exploit (sub)structures that are not too intertwined. The following definitions make these ideas precise:

\begin{definition}[Subpedigrees]
Let $W \subset V_{\cP}$ denote a subset of nodes of pedigree $\cP$.
The subgraph of $(V_{\cP}, E_{\cP})$ induced by $W$
is itself a pedigree, which we call the \textbf{subpedigree} of $\cP$ \textbf{induced by $W$}.
\end{definition}

\begin{definition}[Ancestral pedigrees]
Let $W_k \subset V_k(\cP)$ denote a subset of vertices at level $k$ of a graded pedigree $\cP$.
The subpedigree induced by $W_k \cup \mathsf{anc}(W_k)$ is the (level $k$) \textbf{ancestral subpedigree} of $\cP$ induced by $W_k$.
\end{definition}

\begin{definition}[Descendant pedigrees]
Let $W_k \subset V_k(\cP)$ denote a subset of vertices at level $k$ of a graded pedigree $\cP$.
The subpedigree induced by $W_k \cup \mathsf{desc}(W_k)$ is the (level $k$) \textbf{descendant subpedigree} of $\cP$ induced by $W_k$.
\end{definition}

\begin{definition}[Tree pedigrees] \label{def:tree_pedigree}
A pedigree $\cP$ that has no undirected cycles (when the directions of the edges in $E_{\cP}$ are ignored) is called a \textbf{tree} pedigree.
\end{definition}

Note that coupled pedigrees can have edges of multiplicity two, though only in the case where two siblings form a coupled node, which a rare structure in our generative model. In coupled pedigrees, we consider a double edge to be an undirected cycle of length two. Hence, a tree pedigree consists entirely of simple or multiplicity 1 edges.

As we demonstrate (\textit{e.g.}~\cref{lem:nonsiblings_symbols}), coupled tree pedigrees exhibit a type of correlation decay between blocks that enable us to perform inference on the structure.
In contrast, non-tree coupled pedigrees correspond to pedigrees with inbreeding, which can arise in nature and appear in our probabilistic model as well. \cref{sec:examples} illustrates examples of such structures. These types of structures introduce challenges for performing inference under our generative model.

\subsection{Siblings in a pedigree} \label{sec:siblings}

Note that siblinghood is a transitive relationship: if $u,v$ are siblings and $v,w$ are siblings, then so are $u,w$.
As alluded to in \cref{sec:triples}, it is important to look at these relationships in \textit{triplets}. We now detail how one can encode this information as a \textit{3-uniform hypergraph}.

\begin{definition}
A \textbf{3-uniform hypergraph} is a pair $(V,E)$ of vertices and a multiset of edges, so that each edge is an unordered triple $\{u,v,w\}$ of vertices in $V$.
\end{definition}

\begin{definition}
Let $\cP$ be a \underline{coupled pedigree} of depth $T$ (each non-extant node is a set of a pair of individuals).
The \textbf{siblinghood hypergraph} $G_{k}$ of $\cP$ at level $k > 0$ is the 3-uniform hypergraph that describes the three-way sibling relationships of its level-$k$ members.
For every triple $e = \{c_1,c_2,c_3\}$, the edge multiplicity $n(e; G_k)$ is
\[
    n(e; G_k) = \begin{cases}
        0 & \text{if } \not\exists \; (u_1,u_2,u_3) \in c_1 \times c_2 \times c_3 \text{ such that $u_1, u_2, u_3$ are siblings} \\
        1 & \textit{if } \exists \text{ unique } (u_1,u_2,u_3) \in c_1 \times c_2 \times c_3 \text{ such that $u_1, u_2, u_3$ are siblings} \\
        2 & \text{else}
    \end{cases}
\]
The siblinghood hypergraph $G_0$ is defined similarly, by considering each extant individual $u$ as a degenerate (cardinality 1) couple $c_u = \{u\}$ and applying the above definition (Each hyperedge appears zero or once, never twice).
\end{definition}

Recall that a \textbf{clique} in a 3-uniform hypergraph is a collection of vertices such that all possible triplets form an edge.
The next statement is an observation that follows from the definition of $G_k$ and the transitivity of siblinghood.
\begin{proposition}
\label{prop:sibling_container}
If $c_1, \ldots, c_{m}$ are level-$k$ couples that respectively contain individuals $u_1, \ldots, u_m$ which are siblings, then $c_1, \ldots, c_m$ form a clique in $G_k$.
\end{proposition}


\medskip 

\subsection{Probability Tools}
We denote a Poisson distribution with mean $\lambda$ as $\poisson(\lambda)$.  We use some basic tools from probability theory in our proof. 
The first is referred to in literature as \emph{Poisson thinning}, see \textit{e.g.} \cite{thinning}.
\begin{proposition}[Poisson Thinning]
\label{prop:pois-thin}
Let $N \sim \poisson(\lambda)$, and let $X_1, X_2, \ldots$ be iid Ber($p$) random variables that are independent of $N$. Then $X = \sum_{i=1}^{N} X_i$ is $\poisson(\lambda p)$-distributed.
\end{proposition}

Second, we recall that sums of Poisson random variables are themselves Poissons:
\begin{proposition} \label{prop:pois-add}
Fix $N > 0$ and let $X_1, X_2, \ldots, X_N$ be iid $\poisson(\lambda)$ random variables. Then $X = \sum_{i=1}^{N} X_i$ is $\poisson(\lambda N)$-distributed.
\end{proposition}

Third, we will invoke the following standard variants of the Chernoff--Hoeffding bounds for sums of Bernoulli random variables:
\begin{theorem}[Bernoulli tail probability (Chernoff--Hoeffding Bounds)]
Let $X = \sum_{i=1}^{N} X_i$ be a sum of independent Bernoulli$(p_i)$ random variables.
Let $\mu = \E[X]$.
Then
\begin{align*}
    \Pr(X > (1+\delta)\mu) &\leq \exp \left(-\frac{\delta^2}{2+\delta}\mu \right) \\
    \Pr(X < (1-\delta)\mu) &\leq \exp \left(-\frac{\delta^2}{2} \mu\right) \\
    \Pr( |X - \mu| > \gamma N ) &\leq 2\exp( - 2 N \gamma^2 ). 
\end{align*}
\end{theorem}

Lastly and in the same spirit as the Chernoff--Hoeffding bound, we will use the fact that Poisson distributions also have sub-exponential tails.
\begin{proposition}[Poisson tail probability] \label{prop:poisson-tail}
Let $X \sim \poisson(\lambda)$. Then for any $x > 0$, we have
\[ \Pr(|X - \lambda| \geq x) \leq 2 \exp\left( -\frac{x^2}{2(\lambda + x)} \right) \]
\end{proposition}
For a proof, refer to Chapter 2 of \cite{pollard2015mini}.

\section{Structure of Poisson Pedigrees} \label{sec:structure}
\subsection{Model Description}
\label{sec:model}
We now describe our simple model for generating a population and its genetic data. The model is best viewed in two stages. In the first stage, we generate the population as well as the pedigree topology $\cP_{\mathsf{indiv}}$ on these individuals, and in the second stage, we generate the genetic data given this pedigree structure. Note that the random individual pedigree $\cP_{\mathsf{indiv}}$ constructed below is \textbf{graded, monogamous,} and \textbf{complete}.

\medskip

\noindent \textbf{Part I: Pedigree topology}
\begin{enumerate}
\item To generate $\cP_{\mathsf{indiv}}$, start with $N_T = N$ founding individuals in $V_T$ and make an arbitrary maximum matching of these individuals to create a set of mated couples. For each couple, generate an independent $\text{Pois}(\alpha)$ number of children, where $\alpha > 0$ is a fixed parameter throughout the entire pedigree. These newly generated individuals form the nodes in $V_{T-1}$.
\item Repeat the above process to generate the individuals in $V_{T-2},\dots,V_{0}$.
\end{enumerate}

Once we have the population and pedigree structure as above, we generate the genetic data in the following manner. 

\medskip

\noindent \textbf{Part II: Inheritance procedure}
\begin{enumerate}
\item Each individual $u$ in $\cP_{\mathsf{indiv}}$ has a length-$B$ string $\sigma_u$ ($u$'s \textbf{gene sequence}). The string's indices are referred to as \textbf{blocks}.
\item For each founding individual $u$ in $V_T$ and for each block $b \in [B]$, each $\sigma_u(b)$ is drawn i.i.d. uniformly from an alphabet $\Sigma$.
For our model, $\Sigma$ is an infinite-sized alphabet: we simply require that each block of each founder has a unique symbol.
\item Every other individual $v$ in the population has exactly two parents $f$ and $m$.
Conditioned on $\sigma_f$ and $\sigma_m$, independently over $[B]$, the $i$th block of $v$ copies $\sigma_f(i)$ with probability $0.5$ and $\sigma_m(i)$ with probability $0.5$. 
\end{enumerate}




\begin{remark}
    We adopt the following conventions in the remainder of the paper.
    
    \begin{enumerate}
    \item We let $\cP$ denote the \textbf{coupled pedigree} induced (see \cref{def:coupled-pedigree}) by the randomly generated individual pedigree $\cP_{\mathsf{indiv}}$ constructed in Part I above. 
    
    \item We use the term \textbf{coupled node}, or simply \textbf{node} when the context is clear, to refer to a vertex of $\cP$. We use the term \textbf{individual} to refer to an element of $\cP_{\mathsf{indiv}}$ contained in a coupled node of $\cP$. Unless otherwise explicitly noted, parent-child relationships are taken according to the structure of the coupled pedigree $\cP$. That is, given $u, v \in \cP$ we use the phrase, ``$u$ \textbf{is a child of} $v$,'' to mean that the couple $u$ contains an individual who is an offspring of the mated couple $v$. Finally, we say that coupled nodes $u, v \in \cP$ are \textbf{siblings} if $u$ and $v$ contain individuals who are siblings in $\cP_{\mathsf{indiv}}$. 
    
    \item $\Pr$ denotes the probability measure over the randomly generated pedigree $\cP$ as well as the random inheritance procedure. 
    \end{enumerate}
\end{remark}

To given an example of our terminology, there are two individuals in a non-extant coupled node. Each individual is a vertex of $\cP_{\mathsf{indiv}}$, and together they form a coupled node, which is a vertex of $\cP$. Note that as an artifact of our definitions, extant individuals are both coupled nodes \textit{and} individuals in $\cP$. Moreover extant nodes have exactly one parent in $\cP$ given by the coupled node containing the individuals comprising that extant individuals biological parents, as determined by our generative model. 

To further emphasize the previous remark, recall that by the discussion in \cref{sec:definitions}, there is a unique correspondence between coupled pedigrees and individual pedigrees. Hence, it suffices to give a (partial) reconstruction $\hat \cP$ of $\cP$ to (partially) reconstruct the original individual pedigree $\cP_{\mathsf{indiv}}$. Thus the content of our main result \cref{thm:main-formal} and the remainder of this paper primarily work with the coupled pedigree $\cP$.  




\paragraph{Parameters:}
For convenience, we collect the various parameters of interest here. 

\begin{center}
\begin{tabular}{ |c|l|c| } 
\hline
\multicolumn{1}{|c}{\bfseries Parameter} & \multicolumn{1}{|c}{\bfseries Description} & \multicolumn{1}{|c|}{\bfseries Value} \\
\hline
 $N$ & Size of founding population & \\ 
 \hline
 $B$ & Number of blocks for each individual & $\Theta(\log(N))$ \\  
 \hline
 $\alpha$ & Expected \# of children per couple & $\Theta(1)$\\
 \hline
 $T$ & Number of generations in population & $\eps \log(N)$, $\eps = O(1/\log(\alpha))$\\ 
 \hline
 $|\Sigma|$ & Size of block alphabet & $\infty$\\
 \hline 
\end{tabular}
\end{center}
We set $B = O(\log(N))$ for a sufficiently large constant. The expected number of children per couple, $\alpha$, will be set to a sufficiently large constant that is at least 3. Finally, the number of generations $T$ will be set to $\eps \log(N)$, where $\eps > 0$ is sufficiently small with respect to $1/\log(\alpha)$.


\subsection{Concentration bounds and upper bounds on inbreeding}
In this section we quantify the degree of inbreeding in $\cP$. To do so, we first describe an alternative description of our generative model. An equivalent procedure for constructing the coupled pedigree structure $\cP$ is to (1) sample the generation sizes according to Poisson random variables with appropriate parameters, (2) pair up individuals in each generation at random into coupled nodes, and (3) have coupled nodes choose two parent coupled nodes at random from the previous generation. This is described formally below.

\begin{lemma}
\label{lem:equivalent-generation}
The (coupled) pedigree $\cP$ described in~\cref{sec:model} can be equivalently viewed as follows:
\begin{enumerate}
\item Let $N_T := N$ be the size of the founding population. For $i$ from $T$ to $1$: Let $N_i' \stackrel{def}{=} \lfloor N_i / 2 \rfloor \cdot 2$ be the number of individuals in couples, and sample $N_{i-1} \sim \poisson(\alpha N_{i}' / 2)$. 
\item For each level $i$, match the individuals at level $i$ randomly, leaving out a single individual if $N_i$ was odd. 
\item For each level $i$, sample a vector $\vec{v} \in [N_i'/2]^{N_{i-1}}$ from a Multinomial distribution with parameters
\[ (N_{i-1}, (2/N_i', \ldots, 2/N_i')).\]
For any $k \in [N_i' / 2]$, the set of coordinates $\{j:v_j = k\}$ are interpreted as children of the $k^{th}$ couple at level $i$ (and are therefore siblings at level $i-1$).
\item Convert the resulting pedigree on individuals from steps 1--3 to a coupled pedigree $\cP$. 
\end{enumerate}
\end{lemma}

\begin{proof}
The number of vertices at each level in the statement of~\cref{lem:equivalent-generation} is the same as the model in~\cref{sec:model}. This follows by induction. The number of founding vertices $N$ is the same in both models. In the model in~\cref{sec:model}, the number of individuals at level $i-1$ is distributed as $\sum_{j=1}^{N_i'/2} X_j$, where the $X_j$ are iid $\poisson(\alpha)$ and $N_i'$ is the number of individuals at level $i$ that are matched. The value of this sum is distributed as $\poisson(\alpha N_i'/ 2)$ (due to \cref{prop:pois-add}), the same as in the statement~\cref{lem:equivalent-generation}.

The random matching in Step 2 of~\cref{lem:equivalent-generation} is the same as the matching in~\cref{sec:model}. 

The final step in the process above assigns individuals in $V_{i-1}$ to parents in $V_i$ by sampling a vector $\vec{v}$ of length $N_{i-1}$ with entries in $[N_i'/2]$ from a multinomial distribution and assigning individuals to parents based on these labels. Indeed, if we look at the number of children of a fixed couple (say, the $j^{th}$ couple in $V_i$), this is distributed as 
$\binomial(X, 2/N_i')$, where $X \sim \poisson(\alpha N_i'/ 2)$.
By Poisson thinning (\cref{prop:pois-thin}), this distribution is simply
$ \poisson(\alpha)$,
which is exactly the distribution of the number of children of the $j^{th}$ couple in~\cref{sec:model}.
\end{proof}


Next we use tail bounds on Poisson random variables to show that the sizes of each level are well-concentrated with high probability, assuming a sufficiently large size of the initial population. Recall that $N_i$ denotes the number of \textit{individuals} in generation $i$. 

\begin{lemma}[Concentration of generations]
\label{lem:concentration}
Fix $\delta$ such that $0 < \delta < \alpha/2 - 1$, and suppose that the founding population size $N$ is at least $\alpha / \delta + 1$. 
Then, for some constant $C_1 = C_1(\delta)$, with probability at least 
$1 - T \exp(-C_1 \alpha N)$
we have that, for all $i \in \{0, \ldots, T-1\}$
\begin{equation} \label{eq:concentration}
(\alpha/2 - \delta) N_{i+1} \leq N_i \leq (\alpha/2 + \delta)\cdot N_{i+1}.
\end{equation}
\end{lemma}

\begin{remark}
An immediate corollary of this result is that 
\begin{equation} \label{eq:concentration-1}
(\alpha/2 - \delta)^i \cdot N \leq N_{T-i} \leq (\alpha/2 + \delta)^i \cdot N
\end{equation}
for each $i \leq T$ with high probability.
\end{remark}

\begin{proof}[Proof of \cref{lem:concentration}]
Our goal is to upper bound the right-hand-side of
\[ 
\Pr[\text{some $N_j$ fails~\cref{eq:concentration}}] \leq \sum_{i = 0}^{T-1} \Pr[N_{i} \, \, \text{fails~\cref{eq:concentration}} \, | \,N_{i+1} \, \, \text{satisfies~\cref{eq:concentration-1}}] 
\]
and so it suffices to show
\[ 
\Pr[N_i \, \, \text{fails~\cref{eq:concentration}} \, | \, N_{i + 1} \, \, \text{satisfies~\cref{eq:concentration-1}}] \leq 2\exp(-\Theta(\alpha^2 (N-1) / (\alpha + \delta))).
\] 

Consider fixing the number of individuals at level $i+1$ to be an arbitrary number $N_{i+1}$ satisfying~\cref{eq:concentration-1}. We know that the number of individuals at level $i$ is distributed as $N_{i} \sim \poisson(\alpha N_{i+1}' / 2)$. 
By applying the Poisson tail bound \cref{prop:poisson-tail}, we see that 

\begin{align} 
&\Pr\left[|N_{i} - \alpha N_{i+1}' / 2| > (\delta/2) N_{i+1}' \mid N_{i + 1} \, \, \text{satisfies~\cref{eq:concentration-1}} \right] \\
&< 2\exp\left( \frac{- (\alpha N_{i+1}'/2)^2}{2(\alpha/2 + \delta/2)N_{i+1}'}\right) \nonumber \\
&< 2\exp\left( -\alpha \frac{- N_{i+1}'}{4(1+\delta)}\right) 
\end{align}
We now claim that $|N_{i} - \alpha N_{i+1} / 2| > \delta N_{i+1}$ implies that $|N_{i} - \alpha N_{i+1}' / 2| > (\delta/2) N_{i+1}'$, which follows from the facts that $|N_{i+1} - N_{i+1}'| \leq 1$ and that $N_{i+1} \geq N$ (\cref{eq:concentration-1}). Namely, assume that $N_{i} > (\alpha/2 + \delta) N_{i+1}$. Then $N_i > (\alpha/2 + \delta/2) N_{i+1}'$, since $N_{i+1} \geq N_{i+1}'$. Now assume instead that $N_{i} < (\alpha/2 - \delta) N_{i+1}$. Then
\begin{align*}
    N_{i} &< (\alpha/2 - \delta) N_{i+1} \\
    &\leq (\alpha/2 - \delta) (N_{i+1}' + 1) \\
    &\leq (\alpha/2 - \delta/2) (N_{i+1}')
\end{align*}
where in the last line we use the fact that $(\delta/2) N_{i+1}' \geq (\delta / 2) (N-1) \geq \alpha/2$.

Hence, we get that 
\[
    \Pr[|N_{i} - \alpha N_{i+1} / 2| > \delta N_{i+1} \mid N_{i + 1} \, \, \text{satisfies~\cref{eq:concentration-1}}] 
    \leq 2\exp\left( -\alpha \left[\frac{N-1}{4(1 + \delta)}\right] \right)
\]
where we use the fact that $N_{i+1} \geq N$ since $N_{i+1}$ satisfies~\cref{eq:concentration-1}.
\end{proof}

\begin{remark}[Dependence on $\delta$]
The strategy from this point onwards is to condition on the event from \cref{eq:concentration}. Since this event fails with probability that is exponentially small in $N$, we lose only an additive $\exp(-c_{\delta} \alpha N)$ probability.
\end{remark}





As mentioned in~\cref{sec:examples}, two nodes may have significantly higher amounts of symbol overlap caused by inbreeding in their ancestral pedigree than would be expected given their distance in the pedigree. This can cause us to reconstruct an incorrect pedigree if we attempt to explain the symbol overlap without accounting for inbreeding; for instance, we may see two nodes and think they are siblings, when in reality they are cousins with inbreeding in their family tree (see~\cref{sec:examples} for a detailed example). To formally connect different patterns of inbreeding with the amount of spurious symbol overlap they cause, we introduce the notion of \emph{collisions} in an ancestral pedigree. 
Roughly speaking, triples of coupled nodes with relatively few collisions in their ancestral pedigree do not have many spurious overlaps, which we prove in~\cref{sec:symbol-inheritance}. 
We first define collisions and then bound the number that occur under our probabilistic assumptions in \cref{lem:bound-collisions-main}. We also give an alternative characterization of collisions in \cref{lem:collision_interp} that is useful later.

\begin{definition}[Collisions]
\label{def:collision}
Let $\cP$ denote a coupled pedigree. Fix a subset of nodes $A \subset V_k(\cP)$, where $k \neq T$. If $k > 0$, we say that this collection has $z$ \emph{collisions at level} $k+1$ if the set of parents of $A$ in $\cP$ has size $2|A| - z$.
If $k = 0$, we say that it has $z$ collisions at level $1$ if the set of parents in $\cP$ has size $|A| - z$.
Write
\[
    \mathsf{coll}_{k+1}(A) := (\text{\# collisions at level } k + 1 \text{ in } A)
\]

Extend the notion of collisions to ancestral subgraphs as follows. If we have nodes $u_1, \ldots, u_J \in V_k(\cP)$, the number of collisions between the ancestral subpedigrees $\mathsf{anc}(u_j)$ for $j=1, \ldots, J$ is equal to 
\[
    \mathsf{coll}(u_1, \ldots, u_J)
    :=
    \sum_{i=0}^{T-k-1} \mathsf{coll}_{i+1}(\mathsf{anc}_i(u_1) \cup \cdots \cup \mathsf{anc}_i(u_J))
\]
where $\mathsf{anc}_i(u_j)$ denotes the set of ancestors $i$ levels above $u_j$.
\end{definition}

\begin{lemma}[Ancestral collisions, alternate characterization]
    \label{lem:collision_interp}
    Let $u_1, \ldots, u_J$ denote a set of nodes that are all at the same level. Consider the subpedigree $\mathcal{T} = \text{anc}(u_1, \ldots, u_J)$. Let $k_j$ denote the number of nodes in $\mathcal{T}$ that have outdegree $j$ in the subpedigree $\mathcal{T}$. Then
    \[
    \mathsf{coll}(u_1, \ldots, u_J) = \sum_{j \geq 2} (j - 1) k_j. 
    \]
    
\end{lemma}

\begin{proof}
    
    Let $S$ denote a set of nodes at level $i$. Let $k_{ij}(S)$ denote the set of parents of $S$ that have outdegree $j$ in the subpedigree $\mathsf{anc}(S)$. Let $\mathsf{coll}_{i + 1}(S)$ denote the number of collisions that $S$ has at level $i + 1$. Then we claim that
    \begin{equation}
        \label{eqn:collision_i}
        \mathsf{coll}_{i+1}(S) = \sum_{j} (j - 1) k_{ij}(S).
    \end{equation}
    
    This is true by induction on the cardinality of $S$, as we now demonstrate. We prove this assuming that $S$ is a set of non-extant coupled nodes; the case for extant nodes is extremely similar. The base case $|S| = 1$ follows because the unique node $u \in S$ either has two distinct parents, in which case there are no collisions and each has outdegree $1$, or $u$ has a single parent, in which case the number of collisions is $1$ and the parent has outdegree $2$. In both cases~\cref{eqn:collision_i} holds.
    
    For the inductive step, suppose that \cref{eqn:collision_i} is valid for all $S$ with $|S| \leq s$. Now consider $S$ with $|S| = s+1$. Choose an arbitrary $u \in S$ and consider $S' = S \backslash \{u\}$. Observe that by~\cref{def:collision} and induction:
    \begin{align*}
    \mathsf{coll}_{i+1}(S) &= 2|S| - |par(S)| \\
    &= 2|S'| - |par(S')| + 2|\{u\}| - |par(u) \backslash par(S')| \\
    &=  \mathsf{coll}(S') + 2 - |par(u) \backslash par(S')| \\
    &= \sum_{j} (j - 1)k_{ij}(S') + 2 - |par(u) \backslash par(S')|.
    \end{align*}
    Therefore, if $u$ has $\ell \in \{0, 1, 2\}$ parents contained in $par(S')$, then 
    \[
    \mathsf{coll}_{i+1}(S) = \ell + \sum_{j} (j - 1)k_{ij}(S') = 
    \sum_{j} (j - 1) k_{ij}(S), 
    \]
    because each parent of $u$ contained in $par(S')$ increases the degree of some node in $S'$ by $1$. 
    
   
    
    Applying this argument over all levels $i$ to the sets $\cup_{\ell = 1}^J \mathsf{anc}_i(u_\ell)$, we see by~\cref{def:collision} and summing over all levels $i$ that~\cref{lem:collision_interp} holds for coupled nodes.
\end{proof}  

In our model and in light of \cref{lem:equivalent-generation}, a collision between sets $A$ and $B$ intuitively corresponds to a node in $B$ ``choosing'' a parent couple that was already chosen by another node in $A \cup B$. This observation lets us bound the number of collisions between the ancestors of 3 nodes with high probability.


\begin{lemma}[Exponential tail of collisions]
\label{lem:bound-collisions-main}
Fix three nodes $u, v, w \in \cP$ in the same level $k$, and let $c$ be a positive integer. Then
\begin{equation} 
\label{eq:bound-collisions-main}
\Pr[\mathsf{coll}(u,v,w) \geq c] = O\left(\frac{72^c \cdot 2^{2cT}}{N^{c}}\right)
\end{equation}
\end{lemma}
\begin{proof}
We show that the probability on the left-hand-side of~\cref{eq:bound-collisions-main} can be upper bounded by the probability that a binomial random variable with sufficiently small mean is at least $c$, from which the result follows.

We assume that each level has at least $N$ individuals.
This is a high probability event by \cref{lem:concentration} (which actually describes a much stronger situation).
Since we just want an upper bound, we condition such an event and this assumption is made without loss of generality.

Let $S_i := \mathsf{anc}_i(u) \cup \mathsf{anc}_i(v) \cup \mathsf{anc}_i(w)$. 
Note that $|S_i| \leq 3 \cdot 2^i$, regardless of how many collisions have happened underneath it. The distribution of $\mathsf{coll}(\mathsf{anc}_i(u), \mathsf{anc}_i(v), \mathsf{anc}_i(w))$ is equal to a sum of at most $3 \cdot 2^{i+1}$ Bernoulli random variables, two for each node in $S_i$, which are indicator random variables that a parent coupled node selected by some node in $u \in S_i$ is the same as a parent coupled node previously selected by $v \in S_i$~(\cref{lem:equivalent-generation}). Furthermore, each of these indicator random variables is 1 with probability at most $3 \cdot 2^{T+2} / N$, even conditioned on the previously set random variables---indeed, there are only $3 \cdot 2^{i+1} \leq 3 \cdot 2^{T}$ parents selected in total, so there are only this many nodes that can be selected from to cause a collision, and there are at least $\lfloor  N/2 \rfloor \geq N/4$ coupled nodes at level $i+1$. Therefore, the random variable $\mathsf{coll}(S_i)$ is stochastically dominated by $\mathsf{Bin}(3 \cdot 2^{i+1}, 3 \cdot 2^{T+2} / N)$. Let $X_i \sim \mathsf{Bin}(3 \cdot 2^{i+1}, 3 \cdot 2^{T+2} / N)$. Then we get that
\begin{align}
\Pr[\mathsf{coll}(u,v,w) \geq c] &= \Pr[\sum_i \mathsf{coll}_{k+i}(S_i) \geq c] \nonumber \\
&\leq \Pr[\sum_{i=k}^{T-1} X_i \geq c] \nonumber \\
&\leq \Pr[X \geq c] \label{eqn:bound-collisions-1}
\end{align}
where $X \sim \mathsf{Bin}(3 \cdot 2^{T+1}, 3 \cdot 2^{T+2} / N)$. By bounding the binomial tail and noting that we take $N > 144 \cdot 2^{2T}$, (\cref{eqn:bound-collisions-1}) can be bounded by 
\begin{align*}
    \Pr[X \geq c] &\leq \sum_{i=c}^{3 \cdot 2^{T+1}} \binom{3 \cdot 2^{T+1}}{i}  \left(\frac{3 \cdot 2^{T+2}}{N}\right)^i \\
    &\leq \sum_{i = c}^{3 \cdot 2^{T+1}} (3 \cdot 2^{T+1})^{i} \left( \frac{3 \cdot 2^{T+2}}{N} \right)^{i}
    \\
    &\leq 2 \cdot 72^{c} \cdot \frac{2^{2cT}}{N^c}
\end{align*}
\end{proof}
In particular, by union bounding over all triples of nodes in the coupled pedigree $\cP$, we get the following corollary. Note that there are most $(\alpha/2 + \delta)^{T} \cdot N$ nodes in the pedigree when we condition on the high-probability event from \cref{lem:concentration}.
\begin{corollary}
\label{cor:few-colls}
\[ \Pr[\exists u, v, w: \mathsf{coll}(u,v,w) \geq 4] = O\left(\frac{(\alpha/2 + \delta)^{3T} 2^{8T} }{N}\right)\]
\end{corollary}
Since we take the ratio $T / \log(N)$ to be sufficiently small (\cref{sec:model}), the probability of the above event is negligible. Hence, we can assume without loss of generality for the rest of the document that the number of collisions in the ancestral trees of any three nodes is at most 3.

Additionally, by applying \cref{lem:bound-collisions-main} to a single node (repeated three times) and applying linearity of expectation, we can bound the probability that there are many coupled nodes $u$ with collisions in their ancestral pedigrees $\anc{u}$ using Markov's inequality.
We state this as a corollary.
\begin{corollary}
For any $C > 0$,
\label{cor:generic-ancestors}
\[ \Pr\left[ \Big|\{u: \mathsf{coll}(u) \geq 1\}\Big| \geq C (2\alpha + 4\delta)^{T} \right] \leq 72/C\]
as long as $N$ is sufficiently large.
\end{corollary}

\begin{definition}[$d$-Richness]
Fix a pedigree $\cP$, and let $d \geq 3$ be an integer. All extant nodes in $\cP$ are $d$-rich. 
For all $k > 0$, a level $k$-node is \textbf{$d$-rich} if it has at least $d$ children that are $d$-rich.
\end{definition}

\begin{lemma}[Most nodes are $d$-rich] \label{lem:many-rich}
Fix a constant $0 < \tau < 1$, and let $\delta > 0$ as in \cref{lem:concentration}.
As long as $N$ and $\alpha$ are sufficiently large, there exists a constant $C_2 = C_2(\tau, \delta)$ such that with probability $1 - T \exp(-C_2 \alpha N)$, at least $(1-\tau)$ fraction of level-$k$ coupled nodes in $\cP$ are $d$-rich for all $k$.
\end{lemma}

\begin{proof}[Proof of~\cref{lem:many-rich}]

Let the term ``$d$-poor node'' refer to coupled nodes that are not $d$-rich. Let $M_k$ denote the number of coupled nodes at level $k$ in $\cP$. Our goal is to prove an upper bound on the event that there are at least $\tau M_{k+1}$ $d$-poor nodes at level $k+1$, conditioned on the event that there are at least $(1-\tau) M_{k}$ $d$-rich nodes at level $k$.

Let $R_k$ denote the event that there are at least $(1-\tau) M_k$ $d$-rich nodes at level $k$. 
Let $E$ denote the event $(\alpha/2 - \delta)M_{k+1} \leq M_{k} \leq (\alpha/2 + \delta) M_{k+1}$ for all $k$, which occurs with probability
$1-\exp(-C_1 \alpha N)$ by \cref{lem:concentration}. We also condition on the sizes of $M_0, \ldots, M_T$, abbreviating this conditioning as $M_{0:T}$. 

Let $S$ be an arbitrary subset of nodes at level $k+1$ of size $\tau M_{k+1}+1$, and consider the event where $S$ only consists of $d$-poor nodes.
This implies that the number of $d$-rich children of $S$ is at most $(d-1)(\tau M_{k+1} + 1)$.
Let $X_i$ be iid Bernoulli RVs, which represent indicators for the event where the $i$th $d$-rich child chooses at least one of its parents to be in $S$.
Note that $\Pr(X_i = 1) = \left( 1 - \left(1 - \frac{|S|}{M_{k+1}}\right)^2 \right) > \frac{|S|}{M_{k+1}}$.
\begin{align*}
    \Pr(S &\text{ only has $d$-poor nodes} \mid R_{k}, E, M_{0:T}) \\
    &\leq \Pr\left[ \sum_{i=1}^{(1-\tau)M_{k}} X_i \leq (d-1)|S| \;\Bigg|\; M_{0:T} \right] \\
    &\leq \exp \left[ -\frac{(1-\tau) M_k |S|}{2 M_{k+1}} \left( 1 - \frac{(d-1) M_{k+1}}{(1-\tau) M_{k}} \right)^2 \right] &\textit{(Chernoff--Hoeffding Bound)}
\end{align*}


Observe that there are 
$\binom{M_{k+1}}{|S|} \leq \left( \frac{e}{\tau} \right)^{\tau M_{k+1} + 1}$
many choices for $S$.
To apply a union bound, it suffices for $\alpha$ to be large enough so that $\frac{(1-\tau) M_k}{M_{k+1}} \left( 1 - \frac{(d-1) M_{k+1}}{(1-\tau) M_{k}} \right)^2 \approx (1-\tau)\alpha (1 - \frac{d-1}{(1-\tau)\alpha})^2$ looks linear in $\alpha$.
In that case, we obtain a bound of the form
\begin{align*}
    \Pr(\text{at least } &\tau M_{k+1} \text{ $d$-poor nodes at level $k+1$} \mid R_k, E, M_{0:T}) \\
    &\leq \exp\left( -C M_{k+1} \alpha \right).
\end{align*}
Therefore, we may write
\begin{align*}
    \Pr(\text{at least } &(1-\tau) \text{ fraction of $d$-rich at all levels}) \\
    &\geq (1 - e^{-C_1 \alpha N}) \prod_{k=1}^{T} \left( 1 - \exp(-C M_{k+1}  \alpha) \right) \\
    &\geq 1 - \exp(-C_1 \alpha N) - \sum_{k=0}^{T-1} \exp(-C N (\alpha/2 - \delta)^k  \alpha) \\
    &\geq 1 - T\exp(-C_2 \alpha N)
\end{align*}
for an appropriate constant $C_2$ depending only on $\tau$ and $\delta$.

\end{proof}



\begin{lemma}[Cliques have unique parents]
\label{lem:clique_has_unique_parent}
Let $G_k$ denote the siblinghood hypergraph at level $k$. Let $\delta > 0$ be as in \cref{lem:concentration}. 
For a constant $C_3 = C_3(\delta)$, with probability at least $1 - \frac{1}{N}e^{C_3 T \log \alpha} $, 
for all hypercliques $\cC \subset G_k$ with at least one hyperedge, there is a unique node at level $k+1$ that is a parent of every node in $\cC$. We refer this node as the \textbf{parent} of $\cC$.
\end{lemma}
\begin{proof}
    By \cref{prop:sibling_container}, a hyperclique corresponds to a set of coupled nodes that contain a set of mutual siblings, where each couple has at least one of the siblings in it. This establishes that there is a coupled node at level $k+1$ that is at least one parent of every node in $\cC$. In the case where $\cC$ is a hyperclique of extant nodes, we are done: every node in $\cC$ is an individual and has exactly one parent coupled node.

    If $\cC$ is at a higher level, note that there can be at most two parents for $\cC$, as defined above. The reason is that any individual has exactly one parent couple, and since there are only two individuals in a couple, there cannot be three parent couples each with one child in each couple in $\cC$. 
    
    Next we show that if there are two coupled nodes, both of which are parents of $\cC$, then there must be many collisions among the ancestors of $\cC$, and therefore we can rule this out as a low-probability event. Since $\cC$ has at least one hyperedge, we know that $|\cC| \geq 3$. This means that any arbitrary set of three nodes from $\cC$ must have at least $6-2 = 4$ collisions by \cref{def:collision}---but~\cref{cor:few-colls} shows that with probability $O\left( \frac{(\alpha/2 + \delta)^{3T} 2^{8T}}{N} \right)$, this does not occur anywhere in the pedigree.
\end{proof}

\begin{lemma}[Disjointness of maximal cliques]
\label{lem:main_sibling_structure}
Let $G_k$ denote the siblinghood hypergraph at level $k$. For $k = 0$, each extant node is contained in a unique maximal clique, and moreover, the maximal cliques in $G_0$ are vertex disjoint (and thus, also edge-disjoint). For $k > 0$, each node is contained in at most two maximal cliques. Moreover, with probability $1 - \frac{1}{N} e^{C_3 T \log \alpha}$, the maximal cliques in $G_k$ are edge-disjoint.
\end{lemma}
\begin{proof}
    Note that maximal cliques in the siblinghood hypergraph correspond to maximal sets of siblings. The claim for extant nodes is relatively trivial - extants are individuals, and so the maximal sets of siblings partition the set of extant nodes. 
    
    For $k > 0$, since each individual in a coupled node has one pair of parents, a coupled node can have at most two parents. Thus it can be part of at most two sets of siblings. Hence, it is part of at most two maximal cliques.
    
    Finally, we need to establish that the maximal cliques in $G_k$ are edge-disjoint. To do this, it suffices to show that the intersection between any two maximal cliques is less than 3, so there can be no hyper-edge. Indeed, if three nodes that are simultaneously in two maximal cliques, these three nodes would themselves form a clique with two different parents in level $k+1$, which occurs with probability at most $1 - \frac{1}{N} e^{C_3 T \log \alpha}$ by  \cref{lem:clique_has_unique_parent}.
\end{proof}

\subsection{The joint LCA and its uniqueness}
\label{sec:joint_LCA}

The next two lemmas are crucial in \cref{sec:reconstruction} to show that we can accurately collect symbols for accurately reconstructed coupled nodes. Here we define the \textit{joint lowest common ancestor}, which is a special type of LCA for a triple of coupled nodes. 


\begin{definition}
\label{def:joint_LCA}
Let $u, v, w$ denote coupled nodes in $\cP$. We say that $u, v, w$ have a \textbf{joint LCA} $z$ if it holds that $z \in \mathsf{LCA}(u, v, w)$ and there exist distinct children $c_u, c_v, c_w$ of $z$ so that for all $x \in \{u, v, w\}$, $c_x$ is an ancestor of $x$. 
\end{definition}


\begin{lemma}[Joint LCA is unique]
\label{lem:unique_joint_LCA}
Suppose that each triple of coupled nodes in $\cP$ has at most 3 collisions. Further suppose that $u, v, w$ have a joint LCA $z \in \mathsf{LCA}(u, v, w)$. Then $z$ is the unique LCA of $u, v, w$. 
\end{lemma}

\begin{proof}
    For the sake of contradiction, suppose that $u,v,w$ have another LCA $z' \neq z$.
    By the definition of LCA, $z'$ is neither an ancestor nor a descendant of $z$.
    
    If $z'$ is a joint LCA of $u, v, w$, then both $z$ and $z'$ have outdegree $3$ in $\mathsf{anc}(u, v, w)$, which by~\cref{lem:collision_interp} implies that $\mathsf{anc}(u, v, w)$ has at least $2\times(3-1)=4$ collisions.
    
    If $z'$ is not a joint LCA, then $z'$ has outdegree $2$ in $\mathsf{anc}(u, v, w)$. 
    Moreover, there exists a unique lowest node $y \in \desc{z'} \cap \mathsf{anc}(u,v,w)$ that is an ancestor of precisely two nodes in $\{ u, v, w\}$. 
    In particular, $y$ has outdegree at least $2$ in $\mathsf{anc}(u,v,w)$. 
    Observe that the nodes $y, z, z'$ are all distinct. Hence by~\cref{lem:collision_interp}, the number of collisions is at least $2\times(2 - 1) + 1\times(3 - 1) = 4$. 
    
    In either case, $\mathsf{anc}(u, v, w)$ has at least $4$ collisions, which is a contradiction.
\end{proof}

\begin{lemma}[Inheritance paths go through LCA]
    \label{lem:bad_ancestral_paths}
    Suppose that each triple of coupled nodes in $\cP$ has at most 3 collisions. Further suppose that $u, v, w \in \cP$ have an LCA $z$. 
    Let $z'$ denote a strict ancestor of $z$. 
    Then for some $x \in \{ u, v, w \}$, all paths from $z'$ to $x$ in $\cP$ pass through $z$. 
\end{lemma}

\begin{proof}
    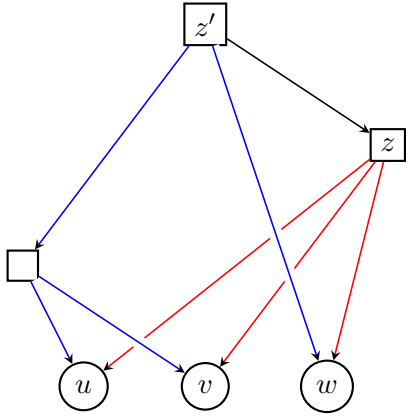
\begin{figure}
    \centering
    \begin{tikzpicture}[scale=0.8,
        > = stealth,
        auto,
        node distance = 3cm, 
        semithick  
    ]
    
	\tikzset{rect state/.style={draw,rectangle,thick,minimum size=4mm}}
	\tikzset{circ state/.style={draw,circle,thick,minimum size=4mm}}
	\tikzset{dashed state/.style={draw,circle,dashed,thick,minimum size=4mm}}
	
	\node[rect state] (Z) at (4,4) {$z$};
	\node[rect state] (Z2) at (1,6) {$z'$};
	\node[rect state] (P) at (-2,2) {};
	\node[circ state] (u) at (-1,0) {$u$};
	\node[circ state] (v) at (1,0) {$v$};
	\node[circ state] (w) at (3,0) {$w$};
	
	\path[->] (Z) edge [draw=red] (u);
	\path[->] (Z) edge [draw=red] (v);
	\path[->] (Z) edge [draw=red] (w);
	\path[->] (Z2) edge [draw=blue] (P);
	\path (Z2) edge [line width=5pt,draw=white] (w);
	\path (P) edge [line width=5pt,draw=white] (u);
	\path (P) edge [line width=5pt,draw=white] (v);
	\path[->] (Z2) edge [draw=blue] (w);
	\path[->] (P) edge [draw=blue] (u);
	\path[->] (P) edge [draw=blue] (v);
	\path[->] (Z2) edge (Z);
    \end{tikzpicture}

    \caption{
    ``Proof-by-picture'' of \cref{lem:bad_ancestral_paths}. 
    }
    \label{fig:bad_ancestral_paths_pic}
    \end{figure}
    
     To draw a contradiction, suppose that for all $x \in \{u, v, w \}$ that $z'$ has a path to $x$ that does not go through $z$. 
     Suppose further, without loss of generality, that $z'$ is the lowest node in $\cP$ that is an ancestor of $z$ and has this property. 
     
     Let $\cT$ denote a spanning tree on $\desc{z} \cap \mathsf{anc}(u,v,w)$ (red edges in \cref{fig:bad_ancestral_paths_pic}). 
     Also select a spanning tree $\cT'$ on the union of all paths from $z'$ to $u, v, w$ that do not go through $z$ (blue edges in \cref{fig:bad_ancestral_paths_pic}). 
     Observe that $z'$ has outdegree at least $2$ in $\cT'$. Since $z'$ also has a path to $z$, then $z'$ has outdegree at least $3$ in $\mathsf{anc}(u,v,w)$. 
     Moreover, $\cT$ has $2$ collisions. Since $z'$ is not contained in $\cT$, we conclude by~\cref{lem:collision_interp} that $\mathsf{anc}(u,v,w)$ has at least $2 + 1 \times (3 - 1) = 4$ collisions. 
     The first terms accounts for the collisions in $\cT$, and the second applies \cref{lem:collision_interp} to $z'$. 
     This is a contradiction.
\end{proof}

Note that by \cref{cor:few-colls}, \cref{lem:unique_joint_LCA,lem:bad_ancestral_paths} hold for all triples $u, v, w \in \cP$ with high probability.




\section{Lemmas that enable reconstruction}

In this section, we prove bounds on ``overlap statistics'' previously explored in \cref{sec:techniques}.
Since we now have switched to talking about coupled pedigrees, we re-define its notion now.


\begin{definition}[Diploid blocks] \label{def:diploid}
Let $\iP$ induce the coupled pedigree $\cP$.
Given (haploid) gene sequences $(\sigma_u)_{u \in V(\iP)}$, we associate with each non-extant couple $v = \{v_1, v_2\}$ node a \textbf{diploid sequence} $\sigma_v$ defined in terms of each block $b$ as a multiset $\sigma _v(b) :=  \sigma_{v_1}(b) \cup \sigma_{v_2}(b)$.
Each extant node's block is thought of as a singleton set.
\end{definition}

\begin{definition}[Diploid overlap] \label{def:diploid-overlap}
Three diploid sequences $\sigma, \sigma', \sigma''$ \textbf{overlap} in block $b$ if
\[
    \sigma(b) \cap \sigma'(b) \cap \sigma''(b) \neq \emptyset.
\]
The term \textbf{fraction of mutual overlaps} between coupled nodes $u, v, w$ in refers to the statistic
\[
\frac{\# \textrm{ overlapping blocks of } \sigma_u, \sigma_v, \sigma_w}{B}
= \frac{ \left|\{ b \in [B]: \sigma_{u}(b) \cap \sigma_v(b) \cap \sigma_w(b) \neq \emptyset  \}\right| }{B}. 
\]
\end{definition} 

\label{sec:symbol-inheritance}
\subsection{Distinguishing siblings from non-siblings: Coincidence probability bounds}

In this section, we establish the following high-probability separation condition for triples of coupled nodes at the same level:
\begin{itemize}
    \item if $u,v,w$ are mutually siblings, they overlap in at least $1/4$ fraction of blocks.
    \item if $u,v,w$ are not mutually siblings, they overlap in at most $3/16$ fraction of blocks.
\end{itemize}
In order to reconstruct the pedigree, we perform inference on the underlying pedigree structure from the symbols at the extant level. The key step of our reconstruction algorithm is to infer which triples of nodes are mutually siblings based on the overlap between their reconstructed symbols. 
The conditions stated above justify using the number of overlapping symbols in triples as a statistic for determining siblinghood.
The first fact (\cref{lem:siblings_symbols}) is easy to prove. In contrast, the second fact (\cref{lem:nonsiblings_symbols}) is rather non-trivial; we prove it using casework.

\begin{lemma}[Symbol overlap in siblings]
\label{lem:siblings_symbols}
With probability $1 - O(\alpha^{3T} N^3 \exp(-\gamma^2 B))$, the fraction of mutual overlap in symbols between any triple of coupled nodes $u$, $v,w \in \cP$ that are mutually siblings is at least $\tfrac{1}{4} - \gamma$ for any arbitrarily small $\gamma > 0$. 
\end{lemma}
\begin{proof}

It suffices to consider the overlap of the individuals $u_1, v_1, w_1$ in $u, v, w$, respectively, that are siblings, \textit{i.e.}, $u_1, v_1, w_1$ have a common parent in $\cP_{\mathsf{indiv}}$. We claim that the expected fraction of overlap for $u_1, v_1, w_1$ is at least 1/4. Indeed, any individual symbol at the parent (couple) node survives to all three children with probability $1/8$, and there are $2B$ symbols at the parent (one per block per member of the couple). 
The Chernoff--Hoeffding bound gives that for any fixed triple $(u,v,w)$ of siblings, the probability that it has less than $1/4 - \gamma$ mutual overlap is at most $\exp(-\gamma^2 B)$. To be explicit, let $X_i$ denote the indicator of an overlap between $u, v, w$ in block $b$. 

\begin{align*}
    \Pr(\text{average overlap} < 1/4 - \gamma) &= \Pr\left( \frac{1}{B} \sum_{i = 1}^B X_i < 1/4 + \gamma\right) \\
    &= \Pr\left( \frac{1}{B} \sum_{i = 1}^B (X_i - \bE[X_i]) < 1/4 - \bE[X_1] + \gamma \right) \\ 
    &\leq \Pr\left( \frac{1}{B} \sum_{i = 1}^B (X_i - \bE[X_i]) < -\gamma \right) \\
    &\leq 2 \exp( - 2 B \gamma^2 ). 
\end{align*}
In the second line we use that $X_i$ are i.i.d., in the third line we use that the expectation is at least $1/4$, and to finish we apply Chernoff--Hoeffding. A union bound over all $O((\alpha^{T}N)^3)$ triples of siblings yields the result.
\end{proof}

\begin{lemma}[Symbol overlap in non-siblings]
\label{lem:nonsiblings_symbols}
Fix $\gamma > 0$. With probability $1 - O(1/N_T) - O(\alpha^{3T}N^3 \exp(-\gamma^2 B))$, every triple of coupled nodes $u$, $v$, and $w$ that are at the same level but are \textbf{not} mutual siblings share overlap in less than $\tfrac{3}{16} + \gamma$ fraction of their symbols.
\end{lemma}

\subsubsection{Proof of Lemma~\ref{lem:nonsiblings_symbols}}

\begin{remark}
\label{rmk:collision_conditioning}
In this proof, we condition on the high probability event from \cref{cor:few-colls} that all triples $u, v, w$ of coupled nodes have at most $3$ collisions in their ancestral subpedigree $\anc{u, v, w}$. 
\end{remark}

It is clear that if $u,v,w$ are completely unrelated, then their mutual overlap is zero, since we assume an infinite alphabet. If $u, v, w$ have a common ancestor, then typically their ancestral pedigree has two collisions, and all triples have at most three collisions in their ancestral pedigree by our conditioning in \cref{rmk:collision_conditioning}. We refer to triples with three collisions as being \textit{inbred} and think of the extra collision as the \textit{site} of inbreeding, a notion that we later formalize in this section. 

Recall the definition of tree subpedigree (\cref{def:tree_pedigree}), which we refer to simply as a \textit{tree} in what follows. Also recall that an edge of multiplicity 2 in a pedigree is considered to be an undirected cycle of length 2. Thus, a tree subpedigree consists only of simple (multiplicity 1) edges. Our strategy for proving \cref{lem:nonsiblings_symbols} follows the recipe below for casework.

\begin{enumerate}
    \item $u,v,w$ have exactly two LCAs, and the ancestral pedigree of 
    $u, v, w$ is a tree.
    \item $u,v,w$ have exactly one LCA, and the LCA has a cycle above it.
    \item  $u,v,w$ have exactly one LCA, and the ancestral pedigree of 
    $u, v, w$ is a tree.
    \item $u, v, w$ have exactly one LCA, and the ancestral pedigree of $u, v, w$ contains a cycle that is not completely above the LCA. 
    
\end{enumerate}

We now assert that the above cases cover all possibilities; this is proven in the next two claims.


\begin{claim}\label{clm:bounded-collisions}
For $u$, $v$, and $w$ to have a single LCA, their ancestors must have at least 2 collisions.
\end{claim}
\begin{proof}
All three nodes need a common ancestor, which means there are at least 2 collisions are present in $\anc{u, v, w}$. 
\end{proof}

\begin{claim} \label{clm:bounded-LCA}
The nodes $u$, $v$, and $w$ have at most two LCAs, with two LCAs only if $\mathsf{anc}(u, v, w)$ has three collisions. Furthermore, if there are two LCAs, then $\anc{u, v, w}$ is a tree pedigree.

\end{claim}
\begin{proof}
By the previous claim, creating a single LCA for three nodes requires 2 collisions in $\anc{u, v, w}$. By definition, one LCA cannot be an ancestor of another LCA. This means there must be at least one more collision in $\anc{u, v, w}$ to create the second LCA, bringing the total number of collisions required in $\anc{u, v, w}$ to three. This immediately yields the final part of the claim by \cref{rmk:collision_conditioning}. 

To establish that there are at most two LCAs, suppose we add a third LCA. Then by the same argument, this LCA cannot be an ancestor of either of the two other LCAs, and so there must be another collision to explain it. This leads to four collisions among the ancestors, which we have ruled out.
\end{proof}

We now upper bound the expected overlap between $u$, $v$ and $w$ by doing the above casework on the structure of their ancestral pedigrees. We simply upper bound the expected overlap, relying on the independence of inheritance in the different blocks so that we can apply a Chernoff--Hoeffding bound.

\begin{lemma}[Case 1: exactly two LCAs] \label{lem:bLCA-case1}
Suppose that $u$, $v$, and $w$ have exactly two LCAs. 
Then the expected fraction of mutual overlap is at most $1/8$.
\end{lemma}
\begin{proof}
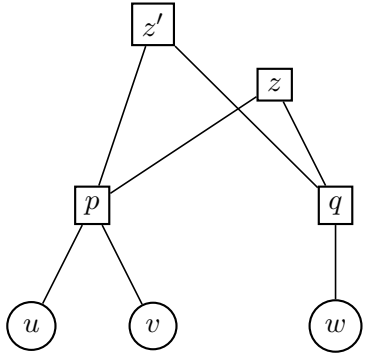
\begin{figure}
    \centering
    \begin{tikzpicture}[scale=0.8,
        > = stealth,
        auto,
        node distance = 3cm, 
        semithick  
    ]
    
	\tikzset{rect state/.style={draw,rectangle,thick,minimum size=4mm}}
	\tikzset{circ state/.style={draw,circle,thick,minimum size=4mm}}
	\tikzset{dashed state/.style={draw,rectangle,dashed,thick,minimum size=4mm}}
	\node[rect state] (Z') at (-1,5) {$z'$};
	\node[rect state] (Z) at (1,4) {$z$};
	\node[rect state] (P) at (-2,2) {$p$};
	\node[rect state] (Q) at (2,2) {$q$};
	\node[circ state] (u) at (-3,0) {$u$};
	\node[circ state] (v) at (-1,0) {$v$};
	\node[circ state] (w) at (2,0) {$w$};
	
	\path (Z') edge (P);
	\path (Z') edge (Q);
	\path (Z) edge (P);
	\path (Z) edge (Q);
	\path (P) edge (u);
	\path (P) edge (v);
	\path (Q) edge (w);
    \end{tikzpicture}
    \caption{The topologies of \cref{lem:bLCA-case1} with two LCAs. Others are obtained by swapping the roles of $u,v,w$.}
    \label{fig:case1}
\end{figure}
\cref{fig:case1} illustrates the topology of interest.
First we note that neither of the LCAs can have repeated symbols, since their ancestral pedigrees contain no collisions. 
Consider the ancestral pedigree from $u$, $v$, and $w$ up to any one particular LCA, noting that this pedigree is a tree by \cref{clm:bounded-LCA}. Any configuration containing $u$, $v$, $w$ and their ancestors leading up to that LCA has at least 5 edges, since $u, v, w$ are not mutual siblings. Therefore, the probability that a single symbol propagates from that LCA to all of $u$, $v$, and $w$ is $\leq (1/2)^5 = 1/32$, which yields an expected 1/16 fraction of overlap since there are $2|B|$ symbols at the LCA (since it is a coupled node).
Since there are two such LCAs, the expectation is at most $1/8$.
\end{proof}


In the remaining cases, we assume there is exactly one LCA. Note that any common symbols across $u$, $v$, and $w$ must be present in this LCA---if $u$, $v$, and $w$ inherit a symbol that is not present in this LCA, then by tracing their paths of inheritance for the symbol we can find another LCA. However, this does not guarantee that \emph{all} common symbols in $u$, $v$, and $w$ can be traced back to inheritance from the LCA--- if there is inbreeding, some nodes in $\{u,v,w\}$ can potentially inherit a symbol via an ancestor of the LCA through a path does not go through the LCA, while the rest inherit it from the LCA. 

\begin{lemma}[Case 2: one LCA with cycle above] \label{lem:bLCA-case2}
Suppose that $u$, $v$, and $w$ have exactly one LCA $z$. Furthermore, this LCA has at least one collision in its ancestral pedigree. Then the fraction of mutual overlap is at most $1/8$ in expectation.
\end{lemma}
\begin{proof}
    
	
We know that $u$, $v$, and $w$, must have at least two distinct parents between them that are connected to $z$ (else $z$ would be their parent). This means there are at least two edges in the graph between $z$ and the parents of $u$, $v$, and $w$, and at least three edges between $u$, $v$, and $w$ and their respective parents. 

Since we know there are at most three collisions among the ancestors of $u$, $v$, and $w$, there can be only one collision in the ancestral pedigree of $z$, and the presence of this collision means there are no other collisions in $\anc{u,v,w}$. Therefore, each of the parent couples of $u$, $v$, and $w$ have an individual that is unrelated to $z$, and so there are no repeated symbols within any of the parent couples. So even if the parents were to get $100\%$ overlap in the blocks due to inheritance from $z$, it holds that $u$, $v$, and $w$ inherit at most $1/8$ fraction of these blocks on expectation.

Finally, all common symbols between $u$, $v$, and $w$ must have been inherited from $z$--- if a common symbol was instead inherited by some $x \in \{u,v,w\}$ from some ancestor of $z$, this would create a fourth collision in $\mathsf{anc}(u,v,w)$.
\end{proof}

\begin{lemma}[Case 3: one LCA and $\anc{u, v, w}$ is a tree] \label{lem:bLCA-case3}
Suppose $u$, $v$, and $w$ have exactly one LCA and $\anc{u,v,w}$ is a tree.
Then the fraction of mutual overlap is at most $1/16$ in expectation.
\end{lemma}
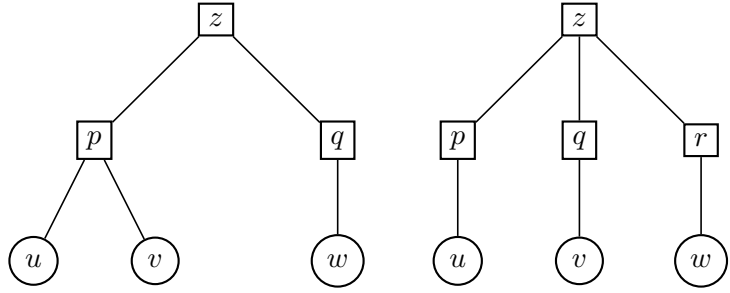
\begin{figure}
    \centering
    \begin{tikzpicture}[scale=0.8,
        > = stealth,
        auto,
        node distance = 3cm, 
        semithick  
    ]
    
	\tikzset{rect state/.style={draw,rectangle,thick,minimum size=4mm}}
	\tikzset{circ state/.style={draw,circle,thick,minimum size=4mm}}
	\tikzset{dashed state/.style={draw,rectangle,dashed,thick,minimum size=4mm}}
	\node[rect state] (Z) at (0,4) {$z$};
	\node[rect state] (P) at (-2,2) {$p$};
	\node[rect state] (Q) at (2,2) {$q$};
	\node[circ state] (u) at (-3,0) {$u$};
	\node[circ state] (v) at (-1,0) {$v$};
	\node[circ state] (w) at (2,0) {$w$};
	
	\path (Z) edge (P);
	\path (P) edge (u);
	\path (P) edge (v);
	\path (Z) edge (Q);
	\path (Q) edge (w);
    \end{tikzpicture}
    \qquad
    \begin{tikzpicture}[scale=0.8,
        > = stealth,
        auto,
        node distance = 3cm, 
        semithick  
    ]
    
	\tikzset{rect state/.style={draw,rectangle,thick,minimum size=4mm}}
	\tikzset{circ state/.style={draw,circle,thick,minimum size=4mm}}
	\tikzset{dashed state/.style={draw,rectangle,dashed,thick,minimum size=4mm}}
	\node[rect state] (Z) at (0,4) {$z$};
	\node[rect state] (P) at (-2,2) {$p$};
	\node[rect state] (Q) at (0,2) {$q$};
	\node[rect state] (R) at (2,2) {$r$};
	\node[circ state] (u) at (-2,0) {$u$};
	\node[circ state] (v) at (0,0) {$v$};
	\node[circ state] (w) at (2,0) {$w$};
	
	\path (Z) edge (P);
	\path (Z) edge (Q);
	\path (Z) edge (R);
	\path (P) edge (u);
	\path (Q) edge (v);
	\path (R) edge (w);
    \end{tikzpicture}
    \caption{Exhaustive list of topologies from \cref{lem:bLCA-case3}, up to re-labelling of $u,v,w$.
    Each edge represents a path of length $> 1$.}
    \label{fig:case3}
\end{figure}
\begin{proof}
The lack of any cycles in $\anc{u,v,w}$ means that all inheritance of common symbols comes from the lone LCA $z$. Any such union of paths from $z$ to $u$, $v$ and $w$ forms a directed tree with at least five edges; see \cref{fig:case3}. In addition, $z$ has two distinct symbols in every block. Therefore, for any particular symbol the probability that all three of $u,v,w$ inherit it is $\leq (1/2)^5 = 1/32$, which yields an expected fraction of at most $1/16$ overlapping blocks. 
\end{proof}

The final case is the most complicated one to analyze. 


\begin{lemma}[Case 4: one LCA with cycle not completely above]
\label{lem:lca-not-tree}

Suppose $u$, $v$, and $w$ have exactly one LCA and $\anc{u, v, w}$ contains a cycle that does not lie completely above $z = \mathsf{LCA}(u, v, w)$. 
Then the fraction of mutual overlap is at most $3/16$ in expectation.
\end{lemma}

As an aid in proving \cref{lem:lca-not-tree}, it is helpful to first identify the \textbf{``most recent'' inbred} node. We make this notion precise now.

\begin{definition}[Witness]
    We call a node $g \in \anc{u, v, w}$ a \emph{witness to inbreeding} or simply a \emph{witness} if $g$ is the lowest node in $\anc{u, v, w}$ that is part of an undirected cycle.  
\end{definition}

\begin{lemma}[Unique witness]
    \label{lem:witness}
    Under the conditions of~\cref{lem:lca-not-tree}, there exists a unique witness in $\anc{u, v, w}$. Moreover, this witness lies strictly below the LCA $z$.
\end{lemma}

\begin{proof}
    We know that $\mathcal{T} := \anc{u, v, w}$ is not a tree, so there exists a cycle in $\mathcal{T}$. We show that there can only be one cycle. Suppose that there exist two cycles $\mathcal{C}, \mathcal{C}'$ in $\mathcal{T}$. Then we claim that $\mathsf{coll}(u, v, w) \geq 4$. 
    
    Consider a spanning tree $\mathcal{T}'$ of $\mathcal{T}$. Then $\mathcal{T}'$ has two collisions. Moreover, $\mathcal{T}' \cup \mathcal{C}$ contains a single cycle, so we conclude that there exists a node in $\mathcal{T}'$ whose outdegree is increased by one upon adding the edges from $\mathcal{C}$ to $\mathcal{T}'$ (Otherwise, $\mathcal{T}' \cup \mathcal{C}$ would still be a tree). Therefore, by~\cref{lem:collision_interp}, $\mathcal{T}' \cup \mathcal{C}$ has three collisions. By similar reasoning and using that $\mathcal{C} \neq \mathcal{C'}$, we conclude that $\mathcal{T}' \cup \mathcal{C} \cup \mathcal{C'}$ has $4$ collisions. Since $\mathcal{T}' \cup \mathcal{C} \cup \mathcal{C'} \subset \mathcal{T}$, we conclude that $\mathcal{T}$ has at least $4$ collisions. But under our conditioning, no subpedigree has $4$ or more collisions. It follows that in~\cref{lem:lca-not-tree} there is exactly one cycle in $\mathcal{T}$, and thus, exactly one witness. 
    
    To prove the final statement, note that if the witness is located above $z$ in $\anc{u,v,w}$, then the cycle lies completely above $z$.
\end{proof}

\begin{proof}[Proof of~\cref{lem:lca-not-tree}]
Consider $u, v, w$ and the subpedigree $\mathcal{T} = \anc{u,v,w}$ consisting of the ancestors of $u, v, w$. Recall that $z$ is the unique LCA of $u, v, w$. By~\cref{lem:witness}, there is a unique witness $g \in \mathcal{T}$, which is the lowest node in the unique cycle occurring in $\mathcal{T}$. 

\medskip

\noindent
\textbf{Subcase 1}:  $\mathsf{LCA}(u, v) = \mathsf{LCA}(v,w) = \mathsf{LCA}(u,w) = \mathsf{LCA}(u, v, w)$.

Without further loss of generality, suppose that the witness $g$ lies along the path from $u$ to $z$. Then it follows that there is a unique path from $v$ to $z$ in $\mathcal{T}$. Otherwise, there would exist two cycles in $\mathcal{T}$, which is a contradiction as this would lead to $4$ collisions in $\mathcal{T}$. Similarly, there is a unique path from $w$ to $z$ in $\mathcal{T}$. Moreover, $\anc{z}$ is a tree. It follows that the subpedigree $\anc{v, w}$ of the ancestors of $v$ and $w$ is a tree. Observe that $z$ is at least two levels above $v, w$, and by the topology of this subcase, there are at least 4 edges in the tree subpedigree from $z$ to $v$ and $w$. This implies that the expected overlap between $v$ and $w$ is at most $2 \cdot (1/2)^4 = 1/8$. Thus the expected overlap between $u, v, w$ is at most the expected overlap between $u$ and $v$, which is bounded by $1/8$. 

\medskip

\noindent
\textbf{Subcase 2}: Without loss of generality, $\mathsf{LCA}(u, v) \neq \mathsf{LCA}(u, v, w)$. 

\begin{figure}
    \centering
    \begin{tikzpicture}[scale=0.8,
        > = stealth,
        auto,
        node distance = 3cm, 
        semithick  
    ]
    
	\tikzset{rect state/.style={draw,rectangle,thick,minimum size=4mm}}
	\tikzset{circ state/.style={draw,circle,thick,minimum size=4mm}}
	\tikzset{dashed state/.style={draw,rectangle,dashed,thick,minimum size=4mm}}
	\node[rect state] (Z) at (0,6) {$z : \{x,y\}$};
	\node[rect state] (G) at (-1,4) {$g$};
	\node[rect state] (P) at (-2,2) {$p$};
	\node[circ state] (u) at (-3,0) {$u$};
	\node[circ state] (v) at (-1,0) {$v$};
	\node[rect state] (Q) at (1,4) {$q$};
	\node[circ state] (w) at (1,0) {$w$};
	\node[dashed state] (1) at (-2,8) {};
	
	\path (Z) edge (G);
     \path[->] (Z.-140) edge [draw=red] node[above left=0.1mm,draw=none] {$\color{red} x$} (G.100);
     \path[->] (1.-120) edge [draw=red] node[above left=0.1mm,draw=none] {$\color{red} x$} (G.130);
	\path (G) edge (P);
	\path (P) edge (u);
	\path (P) edge (v);
	\path (Z) edge (Q);
	\path (Q) edge (w);
	\path (1) edge[dashed] (G);
	\path (1) edge[dashed] (Z);
    \end{tikzpicture}
    \qquad
    \begin{tikzpicture}[scale=0.8,
        > = stealth,
        auto,
        node distance = 3cm, 
        semithick  
    ]
    
	\tikzset{rect state/.style={draw,rectangle,thick,minimum size=4mm}}
	\tikzset{circ state/.style={draw,circle,thick,minimum size=4mm}}
	\tikzset{dashed state/.style={draw,rectangle,dashed,thick,minimum size=4mm}}
	\node[rect state] (Z) at (0,6) {$z: \{x,y\}$};
	\node[rect state] (G) at (1,4) {$g$};
	\node[rect state] (P) at (-2,2) {$p$};
	\node[circ state] (u) at (-3,0) {$u$};
	\node[circ state] (v) at (-1,0) {$v$};
	\node[circ state] (w) at (1,0) {$w$};
	\node[dashed state] (1) at (2,8) {};
	
	\path (Z) edge (G);
     \path[->] (Z.-50) edge [draw=red] node[circle,above right=0.0mm,draw=none] {$\color{red} x$} (G.100);
     \path[->] (1.-60) edge [draw=red] node[circle,below right=0.0mm,draw=none] {$\color{red} x$} (G.45);
	\path (G) edge (w);
	\path (Z) edge (P);
	\path (P) edge (u);
	\path (P) edge (v);
	\path (1) edge[dashed] (G);
	\path (1) edge[dashed] (Z);
\end{tikzpicture}
    \caption{Example of structures being analyzed in the proof of \cref{lem:lca-not-tree}, Subcase 2. Here $\{x, y\}$ depict the symbols of the LCA $z$ in a specific block. The red edges delineate the inheritance events (possibly occurring simultaneously) of a common symbol $x$.}
    \label{fig:subcase2-tree}
\end{figure}
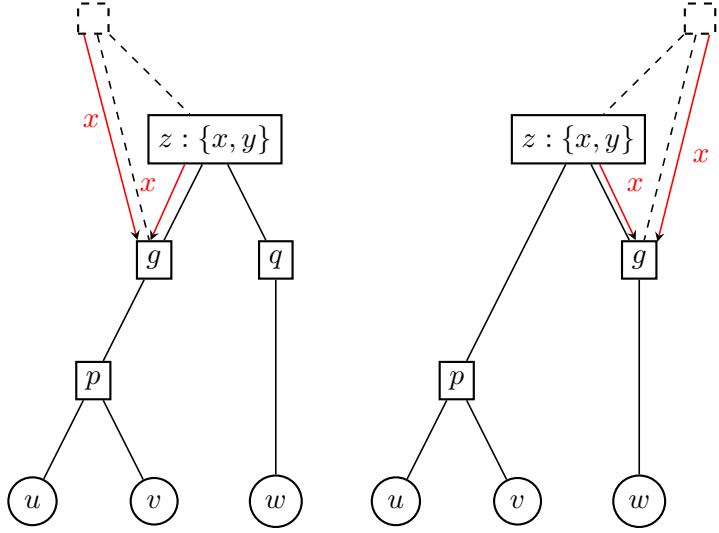

Let $p = \mathsf{LCA}(u, v)$. Either $g$ is on the branch that leads to $u$ and $v$, or it is on the branch that leads to $w$. 
First, suppose that $g$ is on the branch that leads to $u$ and $v$. Then we may further assume $g$ is on the path from $z$ to $p$. For if, say, $g$ is on the path from $p$ to $u$, then $\anc{v,w}$ is a tree, in which case we can argue as in Subcase 1 that the mutual expected overlap between $u, v, w$ is at most $1/8$. 

Therefore, it suffices to consider the cases $g$ is on the path from $z$ to $p$ or $g$ is on the path from $z$ to $w$ (\cref{fig:subcase2-tree}). 
In the first case, the descendants of $g$ form a tree with at least two edges. Moreover, there is a unique node $q$ at the same level as $g$ in $\mathcal{T}$, and this individual is located on the path from $z$ to $w$. Let $\sigma(z) = \{ x, y \}$ denote the (distinct) symbols of $z$ in a given block. By these facts, symmetry, and conditional independence of inheritance,
\begin{align*}
&\Pr[ \sigma(u) \cap \sigma(v) \cap \sigma(w) \neq \emptyset ] \\
&\leq
2\Pr[ \sigma(g) = \{ x,x \}, \, x \in \sigma(u) \cap \sigma(v) ]\Pr[ x \in \sigma(q), x \in \sigma(w)  ] \\
&\quad+ 2\Pr[\sigma(g) = \{ x,y \}, x \in \sigma(u) \cap \sigma(v) ] \Pr[x \in \sigma(q),   x \in \sigma(w)] \\
&\leq 2 \times \left( \frac{1}{4} \times 1 \right) \times \left( \frac{1}{2} \times \frac{1}{2} \right) + 2 \times \left( \frac{1}{2} \times \frac{1}{4} \right) \times \left(\frac{1}{2} \times \frac{1}{2}\right)\\
&= \frac{3}{16}. 
\end{align*}
The second line includes a factor of $2$ to account for either $x$ or $y$ being passed down to $u,v,w$.
The terms in the third line are ordered to correspond to the events in the two lines above. In particular, we have by conditional independence of inheritance that
\[
\Pr[ \sigma(g) = \{x,x\} ] \leq 1/4
\]
because there are at most $2$ paths from $z$ to $g$, and each has probability at most $1/2$ of passing down $x$. The bound 
\[
\Pr[ \sigma(g) = \{x,y\} ] \leq 1/2
\]
holds similarly.

Now suppose that $g$ is on the path from $z$ to $w$. Then
\begin{align*}
    \Pr[ \sigma(u) \cap \sigma(v) \cap \sigma(w) \neq \emptyset ] &\leq 2 \Pr[ x \in \sigma(u) \cap \sigma(v) ] \Pr[ x \in \sigma(w) ] \\
    &\leq 2 \cdot \frac{1}{8} \cdot \frac{3}{4}
    = \frac{3}{16}. 
\end{align*}

Above, we used the fact that tree pedigree from $z$ to $u, v$ has at least $3$ edges. We also used the fact 
\[
    \Pr[ x \in \sigma(w) ] \leq \frac{3}{4},
\]
which holds because there are at most two paths to $w$ from $z$, each path has probability at least $1/2$ of not passing down $x$, and so by conditional independence of inheritance, the probability that both paths do not pass down $x$ is at least $1/4$. 
\end{proof}

\textbf{Finally}, to finish the proof of \cref{lem:nonsiblings_symbols} using Lemmas~\ref{lem:bLCA-case1}, \ref{lem:bLCA-case2}, \ref{lem:bLCA-case3}, and \ref{lem:lca-not-tree}, note that in all four cases the expected overlap between coupled nodes $u, v, w$ is at most $3/16$. Thus, the probability that $u,v,w$ mutually share more than $3/16 + \gamma$ fraction of symbols in all cases is at most $2\exp(-2 B \gamma^2)$ by Chernoff--Hoeffding, similar to the analysis of \cref{lem:siblings_symbols}. Union bounding over all $O((\alpha^T N)^3)$ possible triples gives an $O(\alpha^{3T} N^3 \exp(-B\gamma^2))$ upper bound of the chance that there is some triple with at least $3/16 + \gamma$ overlap. By also ruling out the bad event in \cref{cor:few-colls} (which occurs with probability $O(1/N_T)$), we obtain the desired upper bound.






\subsection{Which ancestors are reconstructible?}

In this section, we characterize nodes that are of importance in our analysis: couples whose history \textit{lacks inbreeding (\textit{e.g.} graph structure is reconstructible using blocks)} and \textit{have ample extant information (\textit{e.g.} blocks are recoverable)}.
We present this in two parts respectively in \cref{def:awesome_node} and \cref{def:new_b_good}.

\begin{definition}[Awesome Node]
\label{def:awesome_node}
Call a node in the pedigree $\cP$ \emph{awesome} if:
\begin{enumerate}
\item It is $d$-rich.
\item It is not an ancestor of any extant node that has a collision within its own ancestral pedigree (including itself).
\end{enumerate}
\end{definition}


\begin{definition}[$b$-goodness]
\label{def:new_b_good}
Let $b \in [B]$ be a specific block. 
Say that a coupled node $v$ in a pedigree $\cP$ is \textbf{$b$-good} if $v$ has at least two sets of three extant descendants $x_1$, $y_1$, $z_1$ and $x_2$, $y_2$, $z_2$ in $\cP$ such that:
\begin{enumerate}
    \item $v$ is a joint LCA of $x_1, y_1, z_1$ and is a joint LCA of $x_2, y_2, z_2$. 
    \item $x_1$, $y_1$, and $z_1$ all have the same symbol $\sigma_1$ in block $b$, and $x_2$, $y_2$, and $z_2$ all have the same symbol $\sigma_2$ in block $b$.
    \item $\sigma_1 \neq \sigma_2$.
\end{enumerate}
We furthermore define every extant node to be $b$-good, for all $b \in [B]$.
\end{definition}

We now deliver the main message of this section: \textit{most nodes have these properties}, given the assumptions of our model (\cref{prop:many-awesome} and \cref{lem:awesome-are-b-good}). 
Therefore, this characterization enables a natural reconstruction algorithm (\cref{sec:reconstruction}).

\begin{proposition}[Many awesome nodes]
\label{prop:many-awesome}
 Let $d > 0$ (as in \cref{def:awesome_node}) be a constant, let $\alpha$ be a sufficiently large constant with respect to $d$, and let $N$ be sufficiently large with respect to both $d$ and $\alpha$. With probability at least $1 - \alpha^{-\Omega(T)}$, in every layer of the pedigree at least $1- 1/d$ fraction of the nodes are awesome.
\end{proposition}
\begin{proof}
    Since $\alpha$ and $N$ are sufficiently large with respect to $d$, we can apply \cref{lem:many-rich} with $\tau = 1/(2d)$ and $\delta = d$. 
    This tells us that at least $1 - 1/(2d)$ fraction of nodes in each layer are $d$-rich with probability $1 - T \exp(-C_2 \alpha N)$, where the constant $C_2 = C_2(1/(2d), d)$ depends only on $d$.
    
    Applying \cref{cor:generic-ancestors} with $C = \alpha^{T}$, there are at most $\alpha^{O(T)}$ nodes at the extant level with collisions in their ancestral pedigree, with probability $1 - \alpha^{-\Omega(T)}$. 
    This means there are at most $2^T \cdot \alpha^{O(T)}$ \emph{ancestors} of these nodes. 
    It follows that the number of nodes that are $d$-rich but not awesome is at most
    $2^T \cdot \alpha^{O(T)}$. This is at most $\frac{N}{2d}$, provided $N$ is sufficiently large with respect to $d$ and $\alpha$ and we take $\eps = T / \log N$ to be small with respect to $1/\log(\alpha)$.
    
    
    The first probability $1 - T\exp(-C_2 \alpha N)$ is exponentially small in $N$, while the second probability $1-\alpha^{-\Omega(T)}$ is exponentially small in $T = \eps \log N$.
    Therefore, the probability of both events occurring simultaneously can be lower bounded by $1 - \alpha^{-\Omega(T)}$, by taking the constant hidden in the $\Omega$ to be slightly smaller than what is found in the previous paragraph.
\end{proof}

\begin{lemma}[Awesome implies $b$-good]
\label{lem:awesome-are-b-good}
Let $d > 0$ (as in \cref{def:awesome_node}) be a sufficiently large constant. With probability $1 - \exp(- \Omega(B))$ over the symbol inheritance process, every awesome coupled node in $\cP$ is $b$-good for at least $99\%$ fraction of blocks $b \in [B]$.
\end{lemma}

The figure ``$99\%$'' is an arbitrary choice for simplification.
It can be replaced by anything arbitrarily close to $1$, which changes the constant factor of $\Omega(B)$ found in the lemma above.
To prove \cref{lem:awesome-are-b-good}, first we need a structural claim about awesome nodes:
\begin{claim}
\label{claim:awesome-tree}
For any awesome coupled node, the subpedigree formed by it and its awesome descendants contains an induced $d$-ary tree that goes down to the extant level.
\end{claim}
\begin{proof}[Proof of~\cref{claim:awesome-tree}]

First, we show that this subpedigree has no undirected cycles within it, which establishes the tree structure. Then, we argue that each node has $d$ children within this subpedigree.

Suppose that there an undirected cycle within this subpedigree. We show that this implies the presence of a collision within the subpedigree, contradicting the awesomeness of all nodes in the subpedigree. Note that there must be a node within this subpedigree with a cycle in its ancestral pedigree - for instance, take the node at the lowest level within the cycle. Applying~\cref{lem:collision_interp} to this awesome node, we see it has a collision among its ancestors, which contradicts condition 2) of \cref{def:awesome_node}.

    Now we establish that each node has at least $d$ children in the subpedigree. An awesome coupled node $v$ has at least $d$ children that are $d$-rich, since it is $d$-rich itself. Furthermore, none of these children have descendants with collisions in their ancestral pedigree, so they are all awesome, which finishes the proof.
\end{proof}

\begin{proof}[Proof of~\cref{lem:awesome-are-b-good}]
    Every awesome coupled node in $\cP$ has exactly $2$ distinct symbols in each block. Indeed, assume for contradiction that there is an awesome coupled node $v$ with a block in which it only has one distinct symbol. Due to the infinite alphabet assumption, we know that we can trace any symbol in a block back to a unique founder. Hence, there must be a collision in the ancestral pedigree of $v$, which is a contradiction with condition 2) of (\cref{def:awesome_node}).
    
    Now we can proceed with showing that every awesome coupled node is $b$-good for $99\%$ fraction of blocks $b \in [B]$. Fix an awesome node $v$ and a block $b \in [B]$. 
    
    We use condition (1) of awesomeness to show that, with probability tending to 1 as $d \to \infty$, there exist two sets of three extant nodes that both have $v$ as a joint LCA, where the first set has a symbol $\sigma_1$ in block $b$, and the second set has a symbol $\sigma_2 \neq \sigma_1$. 
    
    Towards this end, let us follow the inheritance of $\sigma_1$ among an induced $d$-ary tree of awesome descendants, as guaranteed by~\cref{claim:awesome-tree}. The inheritance follows a broadcast process with copy probability $1/2$ on this $d$-ary tree. The probability that the symbol makes it to at least three distinct children of $v$, and this symbol in turn survives to the extant nodes can be expressed as 
    \begin{equation} 
    \label{eq:survival}
    \left(1 - (1/2)^d \left(1 + d + \binom{d}{2} \right)\right) \cdot c_{d, 1/2}
    \end{equation}
    where $c_{d, 1/2}$ refers to the survival probability of percolation on the $d$-ary tree with copy probability $1/2$. The first term refers to the probability that the symbol is inherited by at least 3 of the $d$ awesome children of $v$. Additionally, these three extant nodes have $v$ as an LCA, as they have paths of inheritance from $v$ that do not all intersect at any other node. 
    
    Naturally, \cref{eq:survival} also gives the probability that $\sigma_2$ is similarly inherited. Furthermore, from standard results about Galton-Watson processes (see e.g. \cite{KimAxe15}), we know that as $d \to \infty$, $c_{d, 1/2} \to 1$. Hence, we conclude that \cref{eq:survival} tends to 1 as $d \to \infty$. Thus it follows from the union bound the probability that there exist two sets of three extant nodes that both have $v$ as a lowest common ancestor, the first set has $\sigma_1$ in block $b$, and the second set has $\sigma_2$, also tends to 1 as $d \to \infty$. 
    
    Hence, given a specific block $b$, the probability that an awesome coupled node is $b$-good is at least $0.995$. The high probability of this occurring for all blocks follows from a standard Chernoff--Hoeffding bound.
\end{proof}

\section{Reconstructing the Pedigree}\label{sec:reconstruction}

On the following page, we provide pseudocode for \nameref{algo:reconstruct_all} which is the proposed reconstruction procedure, with details of the inner procedures following it (\nameref{algo:collect-symbols}, \nameref{algo:test-siblings}, and \nameref{algo:construct-parents}). Note that for the first iteration of \nameref{algo:reconstruct_all}, we do not need to collect symbols as the extant genetic data is given to us. Thus we simply test siblinghood at iteration $k = 1$ by using the true gene sequences.

\begin{figure}

\begin{algorithm}[H]
\caption{Reconstruct a depth-$T$ coupled pedigree, given extant individuals $V_0$.}
\begin{algorithmic}[1]
\algolabel{\textsc{Rec-Gen}}{algo:reconstruct_all}
\Procedure{Rec-Gen}{$T, V_0$}
    \State $\hat{\cP} \gets (V = V_0, E = \varnothing)$ \Comment{Extant Pedigree with no edges}
    \For{$k = 1$ to $T$}
        \If{$k > 1$} 
            \ForAll{vertices $v$ in level $k-1$ of $\hat{\cP}$}
                \State $\Call{Collect-Symbols}{v, \hat{\cP}}$
            \EndFor
        \EndIf
        \State $\hat{G} \gets \Call{Test-Siblinghood}{\hat \cP}$
        \State \Call{Assign-Parents}{$\hat{\cP}, \hat{G}$}
    \EndFor
    \State \Return $\hat{\cP}$
\EndProcedure
\end{algorithmic}
\end{algorithm}


\begin{algorithm}[H]
\caption{Empirically reconstruct the symbols of top-level node $v$ in $\cP$.}
\begin{algorithmic}[1]
\algolabel{\textsc{Collect-Symbols}}{algo:collect-symbols}
\Procedure{Collect-Symbols}{$v, \hat{\cP}$}
    \ForAll{blocks $b \in [B]$}
        \Repeat 
            \State \label{algo:line:jointlca} Find extant triple $(x,y,z)$ such that: \par 
            \hskip\algorithmicindent\hskip\algorithmicindent 1) $v$ is a joint LCA of $x, y, z$, \par
            \hskip\algorithmicindent\hskip\algorithmicindent 2) $x$, $y$, and $z$ all have the same symbol $\sigma$ in $b$, and \par
            \hskip\algorithmicindent\hskip\algorithmicindent 3) $\sigma$ is not yet recorded for block $b$ in $v$. 
            \State Record the symbol $\sigma$ for block $b$ in $v$. 
        \Until two distinct symbols are recorded for block $b$, or no such triple exists.
    \EndFor
\EndProcedure
\end{algorithmic}
\end{algorithm}

\begin{algorithm}[H]
\caption{Perform statistical tests to detect siblinghood}
\begin{algorithmic}[1]
\algolabel{\textsc{Test-Siblinghood}}{algo:test-siblings}
\Procedure{Test-Siblinghood}{depth $(k-1)$ pedigree $\hat{\cP}$}
    \State \label{algo:line:mostsymbols} $V \gets \{v \in V_{k-1}(\hat{\cP}) : (\text{\# fully recovered blocks of } v) \geq 0.99|B|\}$
    \State $E \gets \varnothing$
        \ForAll{distinct triples $\{u,v,w\} \subset 2^{V}$ at level $k-1$} 
            \If{$\geq 0.21|B|$ blocks $b$ such that $\hat{s}_{u}(b) \cap \hat{s}_{v}(b) \cap \hat{s}_{w}(b) \neq \varnothing$}
                \State $E \gets E \cup \{u,v,w\}$
            \EndIf
        \EndFor
    \State \Return $\hat{G} = (V,E)$ \label{algo:line:hypergraph} \Comment{3-wise sibling hypergraph}
\EndProcedure
\end{algorithmic}
\end{algorithm}

\begin{algorithm}[H]
\caption{Construct ancestors, given top-level 3-way sibling relationship.}
\begin{algorithmic}[1]
\algolabel{\textsc{Assign-Parents}}{algo:construct-parents}
\Procedure{Assign-Parents}{$\cP, G$}
    \Repeat
        \State $\cC \gets \Call{Any-Maximal-Clique}{G}$
        \State Remove one copy of all hyper-edges in $\cC$ from $G$.
        \State If $|\cC| \geq d$, attach a level-$k$ parent in $\cP$ for all nodes from $\cC$. \label{algo:line:drich}
    \Until no maximal cliques of size $\geq d$ remain in $G$.
\EndProcedure
\end{algorithmic}
\end{algorithm}

\end{figure}

The goal of the rest of this section is to prove the correctness of \nameref{algo:reconstruct_all}.
We now formally state our guarantee:
\begin{theorem}[Main theorem, formal] \label{thm:main-formal}
Let $\hat{\cP}$ be the depth-$T$ coupled pedigree output by the algorithm \nameref{algo:reconstruct_all}, applied to the gene sequences in $V_0(\cP)$.
With probability tending towards $1$ as $N \to \infty$, $\hat{\cP}$ is an induced subpedigree of $\cP$ such that $|V_i(\hat{\cP})| \geq \eta(\alpha) |V_i(\cP)|$ for all levels $i \in \{0, \ldots, T\}$, where $\eta(\alpha) \to 1$ as $\alpha \to \infty$.
The probability is over the randomness of the coupled pedigree $\cP$ and the inheritance procedure with parameters set as in \cref{sec:model}.

\end{theorem}

We define $\eta(\alpha) := 1 - (1/d(\alpha))$ where, for a given value of $\alpha$, $d(\alpha)$ is defined to be the largest value of $d$ such that \cref{prop:many-awesome} holds. Observe that $d(\alpha) \to \infty$ as $\alpha \to \infty$ because \cref{prop:many-awesome} holds for arbitrarily large values $d$. Therefore, $\eta(\alpha) \to 1$ as $\alpha \to \infty$.


We make use of the following high-probability events, provided $\alpha$ is a large enough constant so that $d = d(\alpha)$ satisfies the hypothesis of \cref{lem:awesome-are-b-good}, $N$ is sufficiently large with respect to $\alpha$, the total number of generations is $T = \eps \log N$, where $\eps = O(1/\log \alpha)$, and the gene sequence length is $B = \Omega( \log N)$. 

\begin{proposition}[Key Reductions] \label{prop:reductions}
With probability tending towards $1$ as $N \to \infty$, the pedigree $\cP$ satisfies:
\begin{enumerate}
    \item For each level $k$, each clique of $G_k$ has a single parent (\cref{lem:clique_has_unique_parent}). 
        \label{item:clique_has_unique_parent}
    \item For each level $k$, the maximal cliques of $G_k$ are edge-disjoint, in such a way that each $v \in V_k(\cP)$ is contained in at most two maximal cliques (\cref{lem:main_sibling_structure}). 
        \label{item:main_sibling_structure}
    \item Each triple $u,v,w$ of nodes, has at most 3 collisions (\cref{cor:few-colls}), implying
        \begin{enumerate}
            \item their joint LCA is unique (\cref{lem:unique_joint_LCA}), and
                \label{item:unique_joint_LCA}
            \item all inheritance paths for some node $x \in \{u,v,w\}$ go through the unique LCA (\cref{lem:bad_ancestral_paths}). 
                \label{item:bad_ancestral_paths}
        \end{enumerate}
    \item The fraction of overlap is at least $24.9\%$ for siblings in $\cP$ while for non-mutual siblings it is at most $18.85\%$ (Lemmas~\ref{lem:siblings_symbols}~and~\ref{lem:nonsiblings_symbols}).
        \label{item:siblings_symbols}
    \item For each level $k$, at least $\eta(\alpha)$ fraction of nodes in $V_k(\cP)$ are awesome (\cref{prop:many-awesome}).
        \label{item:many-awesome}
    \item If $u \in V(\cP)$ is awesome, then it is $b$-good for $99\%$ of blocks $b \in [B]$ (\cref{lem:awesome-are-b-good}).
        \label{item:awesome-are-b-good}
\end{enumerate}
\end{proposition}

The ``probability tending towards 1'' portion of \cref{thm:main-formal} can be quantified via a union bound on the probability of failure of any of the events in \cref{prop:reductions}, while the ``$|V(\hat{\cP})| \geq \eta(\alpha) |V(\cP)|$'' guarantee comes from the fact that we recover $100\%$ of the awesome nodes in conjunction with Condition~\ref{item:many-awesome}.
With this as a simplification, we proceed with the proof of \cref{thm:main-formal}.

The upcoming lemma (\cref{lem:extant_reconstruct}) proves the correctness of the very first iteration (depth 1 from depth 0), and therefore serves as the base case.
The inductive step (\cref{lem:level_k_reconstruct}) is presented immediately afterwards.
For the remainder of this section, we write $\hP{k}$ to denote the depth-$k$ reconstructed pedigree after the $k$th iteration of \nameref{algo:reconstruct_all}, ($\hP{0}$ is the depth-$0$ pedigree of all the extant nodes).
In contrast, let $\cP_{k}$ denote the subpedigree of $\cP$ (the ground truth) induced by graded levels $V_0$ up to $V_k$.

\begin{lemma}
\label{lem:extant_reconstruct}
Let $\hat{G}_0$ denote the estimated $3$-regular siblinghood hypergraph for the extant nodes (line~\ref{algo:line:hypergraph} of \nameref{algo:test-siblings}).
Consider the pedigree $\hP{1}$ created by \nameref{algo:construct-parents} applied to $(\hP{0}, \hat{G}_0)$. 
Then there exists an injective homomorphism $\phi:\hP{1} \to \cP_1$ so that the induced subgraph on $\phi(\hP{1})$ is isomorphic to $\hP{1}$.
Moreover, $\phi(\hP{1})$ contains $A_{\leq 1}$, where $A_{\leq 1}$ is the set of awesome nodes at levels $\leq 1$ in $\cP$.
\end{lemma}

\begin{proof}
    
     Let $G_0$ denote the true siblinghood hypergraph on extant nodes with at least two siblings. By Condition~\ref{item:siblings_symbols}, we have that $\hat{G_0} \cong G_0$. 
     Since both graphs have the same set of vertices, we simply write $\hat{G_0} = G_0$.
     
    This gives a natural, explicit characterization of $\phi$.
    For an extant node $v \in V_0(\cP_1)$, define $\phi(v) = v$ so that it is the identity map on the extant. 
    Given couple $\hat{u} \in V_1(\hP{1})$, define $\phi(\hat{u})$ to be the parent couple $u \in V_1(\cP_1)$ of the children of $\hat{u}$. 
    The condition $\hat{G_0} \cong G_0$ implies that at least one such choice for $u$ exists, and moreover by Condition~\ref{item:clique_has_unique_parent}, $u$ is the unique parent.
    
    \textit{$\phi$ is injective}: Let $\hat{u}, \hat{v} \in V_1(\hP{1})$ with $\hat{u} \neq \hat{v}$.
    At the extant level, the maximal cliques in $G_0$ are vertex disjoint by Condition~\ref{item:main_sibling_structure}. 
    Hence, the children of $\hat{u}$ and the children of $\hat{v}$ have empty intersection. Moreover in $\cP_1$, vertex-disjoint maximal cliques have distinct parents. 
    Therefore, $\phi(\hat{u}) \neq \phi(\hat{v})$, as desired. 
    
    \textit{$\phi$ respects edges}: 
    We already know that $(\hat{u}, v) \in E(\hP{1}) \implies (\phi(\hat{u}),v) \in E(\cP_{1})$.
    Now suppose that $(\phi(\hat{u}), v)$ is an edge in $\cP_{1}$ for $\hat{u} \in V_1(\hP{1})$ and $v \in V_0(\hP{1})$. 
    Since $u$ is in the image of $\phi$, it follows that $u$ has at least $3$ children $w, x, y$ that passed the siblings test in our algorithm. 
    If $v$ is one of $w, x, y$, we're done, so suppose not. 
    By Condition~\ref{lem:siblings_symbols}, the extant triples $\{v, w, x\}$, $\{v, x, y\}$, and $\{v, w, y\}$ all have at least $24 \%$ overlap.
    Therefore, $v, w, x, y$ form a clique in $\hat{G_0}$, and line \ref{algo:line:drich} of \nameref{algo:construct-parents} states that $\hat{u}$ is a parent of all four, so $(\hat{u}, v)$ is an edge in $\hP{1}$. 
    
    \medskip
    
    \textit{The image of} $\phi$ \textit{contains the awesome nodes in} $\cP_{1}$: 
    This part is trivially true for the extant nodes, so consider only the awesome nodes $A_1 \subset V_1(\cP_1)$. 
    By definition, any awesome node $u \in V_1(\cP_1)$ is $d$-rich. 
    Since $d \geq 3$, the children of $u$ form a maximal clique of size at least $3$ in $G_0$.
    Therefore, \nameref{algo:construct-parents} creates a parent $\hat{u}$ for these children in $\hP{1}$, which gives the pre-image of $u$.
\end{proof}

\begin{lemma}
\label{lem:level_k_reconstruct}
Let $k \geq 2$ and suppose that we are given $\hP{k-1}$.
Assume that there exists an injective homomorphism $\phi: \hP{k-1} \to \cP_{k - 1}$ which satisfies 
\begin{enumerate}
    \item $\phi\vert_{\hP{0}} \equiv Id$,
    \item $\phi(\hP{k-1}) \subset \cP_{k-1}$ induces a subgraph isomorphic to $\hP{k-1}$, and
    \item $\phi(\hP{k-1})$ contains the awesome nodes in sets $A_0, A_1, \ldots, A_{k-1}$.
\end{enumerate}
Let $\hP{k}$ be the level-$k$ extension of $\hP{k-1}$, via lines 4 through 7 of \nameref{algo:reconstruct_all}.
Then there exists a level-$k$ extension of the map $\phi: \hP{k} \to \cP_{k}$ with the same properties.
\end{lemma}

We prove this in two stages.
The first part (\cref{lem:sibling_graph_reconstruct}) asserts that we reconstruct the sibling relationships correctly, while the latter (\cref{lem:sibling_graph_cliquing}) assures that the cliques of this estimated siblinghood hypergraph are actually the faithful, ``largest possible'' groupings of siblings.

\begin{lemma}
\label{lem:sibling_graph_reconstruct}
    Assume the hypotheses of \cref{lem:level_k_reconstruct}, and let $\hat{G}_{k-1}$ be the estimated siblinghood hypergraph constructed by \nameref{algo:test-siblings}, line~\ref{algo:line:hypergraph}, on input $\hP{k-1}$.
    Then the subgraph of $G_{k-1}$ induced by $\phi(\hat{G}_{k-1})$ is isomorphic to $\hat{G}_{k-1}$, and moreover $\phi(\hat{G}_{k-1})$ contains all of the awesome nodes $A_{k-1}$ at level $k-1$.
\end{lemma}

     The upcoming statements (\cref{clm:awesome_subtree_k}, \cref{clm:symbol_consistency} and \cref{clm:awesome_symbol_recovery}) are pivotal for the proof of~\cref{lem:sibling_graph_reconstruct}.
     
    \begin{definition}
    \label{def:awesome_subtree}
    For an awesome node $u \in \cP_k$, its \emph{awesome subtree} is the subgraph of $\cP_{k}$ that is the union of all paths from $u$ to extant nodes that consist entirely of awesome nodes.
    \end{definition}
    
    \begin{claim} 
\label{clm:awesome_subtree_k} Suppose that there is a reconstruction map $\phi: \hP{k-1} \to \cP_{k - 1}$ satisfying the hypotheses in \cref{lem:level_k_reconstruct}. 
Then for any awesome node $u = \phi(\hat{u}) \in V_{k-1}(\cP_{k-1})$, its awesome subtree $S_u$ satisfies $\phi^{-1}(S_u) = \desc{\hat{u}}$.
\end{claim}

\begin{proof}[Proof of \cref{clm:awesome_subtree_k}]
     Note that Line~\ref{algo:line:drich} of \nameref{algo:construct-parents} ensures that every node in $\hP{k-1}$ is $d$-rich. 
     Since $\phi$ is an injective homomorphism, it follows that every node in $\phi(\desc{\hat{u}})$ is also $d$-rich in $\cP$.
     Furthermore, $u$ being awesome implies that all of its descendants are awesome in $\cP$, since none of its descendants can have collisions in its ancestral pedigree (\cref{def:awesome_node}).
     By the definition of the awesome subtree (\cref{def:awesome_subtree}), it holds that $\phi(\desc{\hat{u}}) \subseteq S_u$. 
    
    For the other direction ($\phi(\desc{\hat{u}}) \supseteq S_u$), let $v \in V_0(\cP)$ be an extant node so that there is a path from $u$ to $v$ consisting only of awesome nodes.
    By condition 3 of \cref{lem:level_k_reconstruct}, all of the nodes along this path are in the image of $\phi$.
    \end{proof} 
    
    \begin{claim}
    \label{clm:symbol_consistency}
    Let $\phi$ be as in \cref{lem:level_k_reconstruct}, and let $u = \phi(\hat{u})$ for some $\hat{u} \in V_{k-1}(\hP{k-1})$. Suppose that in block $b$ the symbols $\hat{\sigma}_1$ and $\hat{\sigma}_2$ are recovered for $\hat{u}$ by applying Algorithm 1 to $\hat{u}$. Then it holds that $u$ also has symbols $\hat{\sigma}_1, \hat{\sigma}_2$ in block $b$.
    \end{claim}
    
    \begin{proof}[Proof of \cref{clm:symbol_consistency}] 
           For $i = 1, 2$, suppose that nodes $x_i, y_i, z_i \in V_0(\hP{0}) = V_0(\cP)$ have the symbol $\hat{\sigma}_i$ in block $b$ and are used by~\nameref{algo:collect-symbols} to recover $\hat{\sigma}_i$ in block $b$ of $\hat{u}$. Recall that $x_i, y_i, z_i$ are all descended from distinct children of $\hat{u}$. 
           Let $\phi(\hat \cP_{k-1})$ induce subpedigree $\cQ$ in $\cP$.
           
           By the hypotheses of \cref{lem:level_k_reconstruct}, $\cQ \cong \hP{k-1}$ and so $u$ must be a common ancestor of $x_i, y_i, z_i$ in $\cQ$.
           By line~\ref{algo:line:jointlca} of \nameref{algo:collect-symbols} and because $\cQ \cong \hP{k-1}$, $\hat{u}$ -- and therefore $u$ -- is their joint LCA.
           With respect to $\cP$, Conditions~\ref{item:unique_joint_LCA} and \ref{item:bad_ancestral_paths} tell us the much stronger condition that $u$ is their only LCA, and that 
           all paths in $\cP$ from any common ancestor of $x_i, y_i, z_i$ to $x_i$ (without loss of generality) must pass through $u$.
           Therefore, if $x_i, y_i, z_i$ all inherit symbols $\hat{\sigma}_i$ in block $b$, the symbol $\hat{\sigma}_i$ must have passed through block $b$ of $u$ via the infinite symbols assumption.
    \end{proof}
    
    \begin{claim}
     
    \label{clm:awesome_symbol_recovery}
    Let $\phi$ be as in \cref{lem:level_k_reconstruct}, and let $u = \phi(\hat{u})$ for some $\hat{u} \in V_{k-1}(\hP{k-1})$.
    Suppose that $u$ is awesome in $\cP$.
    If $u$ is $b$-good and has symbols $\sigma_1, \sigma_2$ in block $b$, then \nameref{algo:collect-symbols} recovers the symbols $\sigma_1$ and $\sigma_2$ for $\hat{u}$ in block $b$.
    \end{claim}
    
     
    
    \begin{proof}[Proof of \cref{clm:awesome_symbol_recovery}]
    By \cref{clm:symbol_consistency}, 
    we only need to show that at least two symbols in block $b$ are reconstructed by \nameref{algo:collect-symbols} applied to $\hat{u}$.
    Note that $b$-goodness implies $\sigma_1 \neq \sigma_2$.
     
    By $b$-goodness of $u$, as in the proof of \cref{lem:awesome-are-b-good}, there is a witnessing triple for each of the $\sigma_i$ contained in the extant of the awesome subtree $S_u$.
    By \cref{clm:awesome_subtree_k}, $\desc{\hat{u}}$ also contains these witnesses.
    Since extant nodes are the exact same in $\cP$ compared to $\hP{k-1}$ by hypothesis 1 of \cref{lem:level_k_reconstruct}, \nameref{algo:collect-symbols} applied to $\hat{u}$ recovers $\sigma_1, \sigma_2$ in block $b$.
    \end{proof}

\begin{proof}[Proof of~\cref{lem:sibling_graph_reconstruct}]
    
    By assumption, $\phi: \hat{G}_{k-1} \to G_{k-1}$ is injective.
    To first see that $\phi$ is a hypergraph homomorphism, let $\hat{u}, \hat{v}, \hat{w} \in V_{k-1}(\hP{k-1})$ be distinct nodes satisfying line~\ref{algo:line:mostsymbols} of \nameref{algo:test-siblings}, and let $u = \phi(\hat{u}), v = \phi(\hat{v})$, and $w = \phi(\hat{w})$ denote their counterparts in $\cP$.
    
    Suppose that $u, v, w$ are not mutually siblings. 
    By Condition~\ref{item:siblings_symbols}, $u, v, w$ have at most $0.1885 |B|$ mutually overlapping blocks. By~\cref{clm:symbol_consistency}, for all $\hat{x} \in \{ \hat{u}, \hat{v}, \hat{w} \}$, the symbols reconstructed for $\hat{x}$ in block $b$ using \nameref{algo:collect-symbols} are a subset of the symbols in block $b$ of $x:= \phi(\hat{x}) \in \{u, v, w\}$. 
    Therefore, $\hat{u}, \hat{v}, \hat{w}$ have mutually overlapping symbols in at most $0.1885 |B|$ blocks. 
    Since $0.1885 < 0.21 $, \nameref{algo:test-siblings} does not place a hyperedge between $\hat{u}, \hat{v}, \hat{w}$ in $\hat{G}_1$.
    
    To conclude that the induced subgraph $\phi(\hat{G}_{k-1})$ is isomorphic to $\hat{G}_{k-1}$, it remains to show that if $u,v,w$ are mutual siblings in $\cP$, then $\{\hat{u}, \hat{v}, \hat{w}\}$ is a hyperedge in $\hat{G}_{k-1}$. 
    Note that $99\%$ of the blocks of $\hat{u}, \hat{v}, \hat{w}$ were recovered by \nameref{algo:collect-symbols} by the definition of $\hat{G}_{k-1}$, and by \cref{clm:symbol_consistency}, the symbols of $\hat{u}, \hat{v}, \hat{w}$ in block $b$ are a subset of the symbols of $u, v, w$, respectively, in block $b$. 
    By Condition~\ref{item:siblings_symbols}, the mutual overlap between the siblings $u, v, w$ is at least $0.249|B|$. 
    Thus, by a union bound on the occurrence of 1\%-fraction of unrecovered blocks, the mutual overlap between $\hat{u}, \hat{v}, \hat{w}$ is at least $(0.249 - 0.03)|B| \geq 0.21|B|$. 
    Therefore, \nameref{algo:test-siblings} constructs a hyperedge on $\hat{u}, \hat{v}, \hat{w}$, as desired. It follows that the induced subgraph $\phi(\hat{G}_{k-1})$ is isomorphic to $\hat{G}_{k-1}$.
    
    Finally, we show that the awesome nodes $A_{k-1}$ are fully contained in $\phi(\hat{G}_{k-1})$.
    By Condition~\ref{item:awesome-are-b-good}, awesome nodes are $b$-good.
    Now apply \cref{clm:awesome_symbol_recovery}, to conclude that \nameref{algo:collect-symbols} reconstructs 99\% of the blocks in each awesome node $u$, so $u \in \hat{G}_{k-1}$ according to Line~\ref{algo:line:mostsymbols} of \nameref{algo:test-siblings}.
\end{proof}

\begin{lemma}
\label{lem:sibling_graph_cliquing}
Let $\mathscr{C}$ denote the maximal (hyper)cliques in the subgraph of $G_{k-1}$ induced by $\phi(\hat{G}_{k-1})$, and let $\mathscr{C}_{\text{algo}}$ denote the (hyper)cliques probed by \nameref{algo:construct-parents} applied to $\hat{G}_{k-1}$. 
Given $\cC \in \aC$, define $\phi(\cC)$ to be the set given by the image of $\cC$ under $\phi$. 
Then $\phi$ is a bijection between $\aC$ and $\sC$. 
\end{lemma}

\begin{proof}
     By \cref{lem:sibling_graph_reconstruct}, the subgraph $H$ induced by $\phi(\hat{G}_{k-1})$ is isomorphic to $\hat{G}_{k-1}$. 
     Hence, it suffices to show that the cliques probed by \nameref{algo:construct-parents} applied to $H$ are precisely the maximal cliques of $H$. 
     Recall that by Condition~\ref{item:main_sibling_structure}, the maximal cliques in $H$ are edge-disjoint, and every node of $H$ is involved in at most $2$ cliques. 
     
    It is helpful to imagine the cliques $\cC_1, \cC_2, \ldots, \cC_M \in \aC$ as being listed out in the same order that they are probed by \nameref{algo:construct-parents}, indexed by timesteps $m = 1, 2, \ldots, M$. Let $H\RP{0} = H$, and let $H\RP{m}$ denote the result of removing the edges of the clique $\cC_t$ from $H\RP{m-1}$. 
    
    We argue that for all $m$, the graph $H\RP{m}$ is a union of edge-disjoint maximal cliques, and any two maximal cliques intersect in at most a single vertex. 
    The base case $m = 1$ is true by Condition~\ref{item:main_sibling_structure}. 
    This holds for $m > 1$ because the above property is preserved when all of the edges are removed from a single maximal clique in $H\RP{m-1}$. 
    Moreover, for all $m$, the maximal cliques in $H\RP{m}$  are the same as those of $H\RP{m-1}$ but with a single maximal clique $\cC_{m}$ in $H\RP{m-1}$ removed. 
    Hence, it also follows by induction that for all $m$, the maximal clique $\cC_m$ in $H\RP{m-1}$ is also a maximal clique in $H$. 
    
    Since \nameref{algo:construct-parents} terminates at the first time $M$ when $H\RP{M}$ has no hyperedges, we conclude that $\cC_1, \ldots, \cC_M$ are \textit{all} of the maximal cliques in $H$, as desired.
\end{proof}

\begin{proof}[Proof of \cref{lem:level_k_reconstruct}]
    We first extend the definition of $\phi$ to level $k$. 
    For $\hat{u} \in V_k(\hP{k})$, we define $\phi(\hat{u}) \in V_k(\cP_k)$ as follows. 
    Let $\hat{\cC} \subset V_{k-1}(\hP{k})$ denote the children of $\hat{u}$. 
    By~\cref{lem:sibling_graph_reconstruct,lem:sibling_graph_cliquing}, $\phi(\hat{\cC})$ is a clique in $G_{k-1}$. 
    Define $\phi(\hat{u})\in V_k(\cP_{k})$ to be the parent of the children of the clique $\phi(\hat{\cC})$ in $\cP$. 
    The map $\phi$ is well-defined at level $k$ because of Condition~\ref{item:clique_has_unique_parent}. 
    It remains to show that $\phi$ is an isomorphism onto its image, and moreover that its image contains all of the awesome nodes at level $k$.
    
    \medskip
    
    \textit{The map} $\phi$ \textit{is injective}: We know this is true for $\phi\big|_{\hP{k-1}}$, so it suffices to consider injectivity of $\phi$ when restricted to the nodes at level $k$ in $\hP{k}$. 
    Let $\hat{u}, \hat{v} \in V_{k}(\hP{k})$ with $\hat{u} \neq \hat{v}$. Let $\cC$ (resp., $\cC'$) denote the maximal clique in $\hat{G}_{k-1}$ that consists of the children of $\hat{u}$ (resp., $\hat{v}$). By~\cref{lem:sibling_graph_cliquing}, $\phi(\cC)$ and $\phi(\cC')$ are distinct maximal cliques in the induced subgraph $\phi(\hat{G}_{k-1})$, and therefore, are contained in distinct maximal cliques in $G_{k-1}$. Distinct maximal cliques in $G_{k-1}$ have distinct parents, so by the definition of $\phi$, we conclude that $\phi(\hat{u}) \neq \phi(\hat{v})$, as desired.
    
    \medskip
    
    \textit{The map} $\phi$ \textit{is edge-preserving}: Suppose that $(\hat{u}, \hat{v})$ is an edge in $\hP{k}$ with $\hat{u} \in V_k(\hP{k})$ and $\hat{v} \in V_{k-1}(\hP{k})$. Consider the maximal clique $\hat{\cC}$ containing $\hat{v}$ in $\hat{G}_{k-1}$. By~\cref{lem:sibling_graph_cliquing}, $\phi(\hat{\cC})$ is a maximal clique in the induced subgraph $\phi(\hat{G}_{k-1}) \subset G_{k-1}$, and by construction of $\phi$, the parent of $\phi(\hat{\cC})$ is $\phi(\hat{u})$. Therefore, the edge $(\phi(\hat{u}), \phi(\hat{v}))$ is in the pedigree $\cP_{k}$. 
    
    Suppose now that the edge $(u, v) = (\phi(\hat{u}), \phi(\hat{v}))$ is in the pedigree $\cP_{k}$. Consider the maximal clique $\cC' \subset G_{k-1}$ containing $v$. By~\cref{lem:sibling_graph_cliquing}, $\cC := \phi^{-1}(\cC') = \{ x \in  \hP{k}: \phi(x) \in \cC' \}$ is a maximal clique in $\hat{G}_{k-1}$.
    By \cref{lem:sibling_graph_cliquing} and the construction in \nameref{algo:construct-parents}, we conclude that the parent of $\hat{v}$ in $\hP{k}$ is mapped to $u$ under $\phi$. 
    By injectivity of $\phi$, this parent is precisely $\phi^{-1}(u) = \hat{u}$. Therefore, $(\hat{u}, \hat{v})$ is an edge in $\hP{k}$. 
    
    \medskip 
    
    \textit{The image of} $\phi$ \textit{contains the awesome nodes in} $\cP_{k}$: It suffices to prove the statement for the awesome nodes at level $k$, which we denote by $A_k$. 
    Suppose that $u$ is an awesome node at level $k$ of $\cP$. 
    By awesomeness, $u$ has at least $d$ awesome children. Let $\cC'$ denote the clique in $G_{k-1}$ given by the awesome children of $u$. 
    By \cref{lem:sibling_graph_reconstruct,lem:sibling_graph_cliquing}, $\cC := \phi^{-1}(\cC')$ satisfies $|\cC| = |\cC'| \geq d$ because all of the awesome children up to level $k-1$ are in the image of $\phi$, by the inductive hypotheses. 
    By \cref{lem:sibling_graph_cliquing}, the maximal clique $\tilde{\cC}$ containing $\cC$ in $\hat{G}_{k-1}$ satisfies that $\phi(\tilde{\cC})$ are all children of $u$. 
    By the definition of \nameref{algo:construct-parents} and $\phi$ at level $k$, we conclude that a parent $\hat{u}$ is constructed for $\tilde{\cC} \supset \cC$ and $\phi(\hat{u}) = u$, as desired. 
\end{proof}

\noindent \textbf{Acknowledgments}
We thank Vishesh Jain for many helpful discussions.  

\bibliographystyle{alpha}
\bibliography{bibliography}
\end{document}